\documentclass{theoretics}

\usepackage{pgf}
\usepackage{multirow}
\usepackage{fancyvrb}
\usepackage{makecell}

\DeclareMathOperator*{\E}{\mathop{\mathbb{E}}}
\newcommand{\can}{ \textnormal{Can}}
\newcommand{\varlabel}{{\rm varlabel}}
\newcommand{\clauselabel}{{\rm clauselabel}}
\newcommand{\N}{\mathbb{N}}

\newcommand{\hkzzbasis}{1.306995}
\newcommand{\pfrac}[2]{\left(\frac{#1}{#2}\right)}
\newcommand{\var}{{\rm var}}
\newcommand{\ppsza}{Paturi, Pudl\'ak, Saks, and Zane}
\newcommand{\ppza}{Paturi, Pudl\'ak, and Zane}
\newcommand{\hkzz}{Hansen, Kaplan, Zamir, and Zwick}
\newcommand{\hsswa}{Hofmeister, Sch\"oning, Schuler, and Watanabe}
\newcommand{\KL}{\textnormal{KL}}

\renewcommand{\root}{\textnormal{root}}
\newcommand{\wcut}{\textnormal{wCut}}
\newcommand{\ppsz}{{\rm ppsz}}
\newcommand{\forced}{ {\rm Forced}}
\newcommand{\cut}{ {\rm Cut}}

\newcommand{\error}{\textnormal{ Error}}
\newcommand{\indeg}{{\rm indeg}}
\newcommand{\heavy}{\textsc{Heavy}}
\newcommand{\heavychild}{\textsc{HCB}}
\newcommand{\bonusheavy}{\textsc{HB}}  
\newcommand{\heavypenalty}{\textsc{HP}}
\newcommand{\setprivileged}{\textsc{Privileged}}
\newcommand{\privlgd}{\textsc{Priv}}
\newcommand{\parentM}{\textsc{ParentM}}
\newcommand{\damage}{\textsc{Damage}}
\newcommand{\benefit}{\textsc{Benefit}}

\newcommand{\benefitSpelledOut}{ \epsilon \gamma^2(r)  (1 - Q_r)^2 P_{r - \delta}^{k-3}}
\newcommand{\damageSpelledOut}{(k-1) (1-r) P_r^{k-2} \delta Q'_r }

\newcommand{\RestateRemark}[1]{{\normalfont\bfseries #1}}
\newcommand{\RestateInit}[1]{\newcommand{#1}{}}
\newcommand{\RestateGo}[1]{\renewcommand{#1}{(Restated)}}

\usepackage{thm-restate}
\SetKwProg{Proc}{procedure}{}{}

\title{PPSZ is better than you think}
\ThCSauthor{Dominik
  Scheder}{dominik.scheder@gmail.com}[0000-0002-9360-7957]

\ThCSaffil{Hochschule Zittau/G\"orlitz, Germany}

\ThCSshortnames{D.\ Scheder}
\ThCSshorttitle{PPSZ is better than you think}

\ThCSyear{2024}
\ThCSarticlenum{5}
\ThCSdoi{10.46298/theoretics.24.5}
\ThCSreceived{Jul 25, 2022}
\ThCSrevised{Oct 19, 2023}
\ThCSaccepted{Feb 4, 2024}
\ThCSpublished{Mar 13, 2024}

\ThCSthanks{An extended abstract of this work has already 
been published~\cite{Scheder2021}, and a full version 
is publicly accessible at~\cite{Scheder2021ECCC}.}

\ThCSkeywords{Satisfiability, NP-complete problems, randomized
  algorithm}

\begin{document}
\maketitle

\begin{abstract}
   PPSZ, for long time the fastest  known algorithm for $k$-SAT, works by 
   going through the variables of the input formula in random order; each variable
   is then set randomly to $0$ or $1$, unless the correct value can be inferred
   by an efficiently implementable rule (like small-width resolution; or being implied by a 
   small set of clauses).
   
   We show that PPSZ performs exponentially better than previously known,
   for all $k \geq 3$. We achieve this through an improved 
   analysis and without any change to the algorithm itself. 
   The core idea is to pretend that PPSZ does not process the variables
   in uniformly random order, but according to a carefully designed distribution. We write
   ``pretend'' since this can be done while running the original 
   algorithm, which does use a uniformly random order.

\end{abstract}

 \section{Introduction}
 
 Satisfiability is a central problem in theoretical computer science. One is given a Boolean formula
 and asked to find a {\em satisfying assignment}, that is, setting the input variables to $0$ and $1$ to make
 the whole formula evaluate to $1$. Or rather, determine whether such an assignment exists.
 A particular case of interest is CNF-SAT, when the input formula
 is in conjunctive normal form---that is, the formula is an AND of clauses; a clause is an OR of literals;
 a literal is either a variable $x$ or its negation $\bar{x}$. If every clause contains at most $k$ literals,
 the formula is called a $k$-CNF formula and the decision problem is called $k$-SAT.
 
 Among worst-case algorithms for $k$-SAT, two paradigms dominate: local search algorithms like
 Sch\"oning's algorithm~\cite{Schoening1999} and random restriction algorithms like PPZ (\ppza{}~\cite{ppz}) 
 and PPSZ (\ppsza{}~\cite{ppsz}).
 Both have a string of subsequent improvements:
\hsswa{}~\cite{HSSW}, Baumer and Schuler~\cite{BaumerSchuler2003},
and Liu~\cite{Liu18} improve Sch\"oning's algorithm.  Hertli~\cite{Hertli2014} and \hkzz{}~\cite{HKZZ}
  improve upon PPSZ.
 
 For large $k$, both paradigms achieve a running time of the form $2^{ n (1 - c/k + o(1/k))}$, 
 where $c$ is specific
 to the algorithm ($c = 1$ for PPZ; $c = \log_2(e) \approx 1.44$ for Sch\"oning; 
 $c = \pi^2 / 6 \approx  1.64$ for PPSZ). Interestingly, the running time of completely different approaches like the polynomial method
 (Chan and Williams~\cite{ChanWilliams})  is also of this form. This gave rise to the 
 Super-Strong Exponential Time Hypothesis (Vyas and Williams~\cite{VyasWilliams}),
which  conjectures that the $c/k$ in the exponent is optimal; for instance,
that a running time of $2^{n ( 1- \log(k) / k)}$ is impossible.\\

This paper presents an improvement of PPSZ. However, it is not an improvement of the algorithm
but of its {\em analysis}. We show that the same algorithm performs exponentially
better than previously known. Informally, PPSZ works by going through the variables
in random order $\pi$; inspecting each variable $x$, it tosses an unbiased coin to determine which value
to assign, unless there is a set of at most $w$ clauses that implies a certain value for $x$.
Take $w=1$ and this is exactly PPZ; take $w = \omega(1)$ and this is PPSZ (the exact
rate by which $w$ grows turns out to be  immaterial for all currently known ways to analyze the 
algorithm). 
Our idea is to pretend that the ordering $\pi$ is not chosen uniformly but from a carefully
designed distribution~$D$. This increases the success probability of PPSZ by some ``bonus'',
which depends on $D$. It seems
surprising that this can be done without actually changing the algorithm, 
but it turns out to be just a straightforward manipulation, which we formally
explain below in (\ref{distorted-jensen}). There is a price to pay in terms of how much $D$ differs from the uniform distribution: 
the success probability  incurs  a penalty of $2^{ - \KL(D||U)}$, where $\KL(D||U)$ is the Kullback-Leibler
divergence from the  uniform distribution $U$  to $D$.
We focus on  Unique-$k$-SAT, where
the input formula has exactly one satisfying assignment. A ``lifting theorem'' by Steinberger and myself~\cite{SchederSteinberger} shows that improving PPSZ for Unique-$k$-SAT automatically 
yields a (smaller) improvement for general $k$-SAT problem (without changing the algorithm).

The idea of analyzing PPSZ assuming some non-uniform distribution $D$ on permutations and paying a price 
in terms of $\KL(D||U)$ is not new. It is explicit in \cite{SchederSteinberger} and implicit in 
\cite{Hertli2011} and  \cite{ppsz}. However, all previous applications use this to deal with the case that 
${\rm sat}(F)$, the set of satisfying assignments, contains multiple elements; furthermore, in
\cite{SchederSteinberger, Hertli2011, ppsz}, the distribution $D$ is defined solely in terms of 
${\rm sat}(F)$ and ignores the syntactic structure of $F$ itself. In particular,
in the special case that $F$ has a unique solution,  $D$ reverts to the uniform distribution.
This paper is the first work that exploits the structure of $F$ itself to define a distribution $D$ on permutations,
and uses this to prove a better success probability for the Unique-SAT case.

 \subsection{Analyzing PPSZ: permutations and forced variables}

 We will now formally describe the PPSZ algorithm. Let $F$ be a formula, $x$ a variable,
 and~$b \in \{0,1\}$. A  formula $F$ {\em implies} $(x = b)$ if every satisfying assignment of $F$
 sets $x$ to $b$. For example, $(x \vee y) \wedge (x \vee \bar{y})$ implies $(x = 1)$ but neither
 $(y = 0)$ nor $(y = 1)$. For an integer $w$, we say $F$ {\em $w$-implies } $(x=b)$ if there 
 is a set $G$ of at most $w$ clauses of $F$ such that $G$ implies $(x=b)$.
 
\paragraph{The PPSZ algorithm with strength parameter $w$.}
 Let $w = w(n)$ be some fixed, slowly growing function. Given a CNF formula $F$
 and a permutation $\pi$ of its variable set $V$, 
 we define $\ppsz(F,w,\pi)$ as follows: go through the variables
 $x_1,\dots,x_n$ in the order prescribed by $\pi$. In each step, when handling a variable $x$,
 check whether $(x=b)$ is $w$-implied by $F$ for some $b \in \{0,1\}$. If so, set $x$ to $b$ (i.e.,
 replace every occurrence of $x$ in $F$ by $b$). Otherwise, set $x$ randomly to $0$ or $1$,
 with probability $1/2$ each. We define $\ppsz(F,w)$ to first choose a uniformly
 random permutation $\pi$ and 
 then call $\ppsz(F,w,\pi)$.
\begin{algorithm}[h]
\caption{PPSZ with fixed permutation $\pi$}\label{algorithm-ppsz-fixed-pi}
\Proc{$\ppsz(F, w, \pi)$}{
   $V := $ the set of variables in $F$\;
   $\beta := $ the empty assignment on $V$\;
   \For{$x \in V$ in the order of $\pi$}{
     \If{there is $b \in \{0,1\}$ such $(x = b)$ is $w$-implied by $F$\label{line-if}}{
      $\beta(x) := b$ \label{line-infer}\;
     }
     \Else{
       $\beta(x) := $ a uniformly random bit in $\{0,1\}$ \label{line-else}\;
     }
      $F := F|_{x = \beta(x)}$\;
   }
   \If{$F$ has been satisfied}{
      \Return{$\beta$}
     }
   \Else{
    \Return{\texttt{failure}}
   }
}
\label{algo-ppsz-fixed-pi}
\end{algorithm} 

\begin{algorithm}[h]
\caption{PPSZ with random permutation $\pi$}\label{algorithm-ppsz-random-pi}
\Proc{$\ppsz(F, w)$} {
   $\pi := $ a random permutation of the variables $V$\;
   \Return{$\textsc{ppsz}(F, w, \pi)$}
}
\label{algo-ppsz-random-pi}
\end{algorithm}

 It should be noted that \ppsza{} in~\cite{ppsz} formulated a stronger version of PPSZ, which tries to 
 infer $(x=b)$ using  bounded-width resolution. 
 A close look at their proof shows that they never use properties of resolution beyond those
 already possessed by $w$-implication.
 We assume that $\alpha = (1,\dots,1)$ is the unique satisfying assignment of $F$. This is purely
 for notational convenience.
\begin{definition}
   Let $\pi$ be a permutation of $V$ and $x$ a variable.  Let $A \subseteq V$ be the set of variables
   coming before $x$ in $\pi$, and let $F' := F|_{A \mapsto 1}$ be the restricted  formula obtained
   from $F$   by setting every variable $y \in A$ to $1$. If $F'$ $w$-implies $(x = 1)$ then 
   we say $x$ is {\em forced under $\pi$} and write $\forced(x, \pi) = 1$; otherwise we say 
   $x$ is {\em guessed under $\pi$} and write $\forced(x,\pi) = 0$.
   Let $\forced(\pi) := \sum_{x \in V} \forced(x, \pi)$.
\end{definition}  

Of course, $\forced(x,\pi)$ depends on $F$ and $w$, as well, so to be 
formally correct, we should write $\forced(F,w,x,\pi)$. Since 
$F$ and $w$ are fixed throughout, we prefer the less
 formal notation $F(x,\pi)$.

\begin{observation}[\cite{ppsz}]
   Suppose we run PPSZ with a fixed permutation $\pi$. Then
   $\ppsz(F,w,\pi)$  succeeds, i.e., finds $\alpha$, with probability exactly $2^{ - n + \forced(\pi)}$.
\end{observation}

Taking $\pi$ to be  a random permutation we get
\begin{align}
\Pr[\ppsz(F) \textnormal{ succeeds}]  & = \E_{\pi} \left[ 2^{-n + \forced(\pi)}\right] 
\label{succ-prob-pre-jensen} \\
  & \geq  2^{ -n + \E_{\pi} [\forced(\pi)]} \ , \nonumber 
   \end{align}
which follows from Jensen's inequality applied to the convex function $t \mapsto 2^t$. We are now
in a much more comfortable position: $\E[\forced(\pi)] = \sum_{x} \Pr[\forced(x,\pi) = 1]$,
and we can analyze this probability for every variable individually. Indeed, this is what \ppsza{}~\cite{ppsz} did:
they showed that $\Pr[\forced(x, \pi)=1] \geq s_k - o(1)$ if $F$ is a $k$-CNF formula with exactly one satisfying assignment. 
 Here $s_k$ is a number defined by the following experiment: let $T_{k-1}^{\infty}$ be the complete rooted
 $(k-1)$-ary tree; pick $\pi: V(T_{k-1}^{\infty}) \rightarrow [0,1]$ at random and delete every node $u$ with 
 $\pi(u) < \pi(\root)$, together with all its 
 descendants. Let $\mathbf{T}$ be the resulting
 tree. Then
 \begin{align}
 \label{definition-sk}
 s_k := \Pr[\mathbf{T} \textnormal{ is finite}] \ .
 \end{align} 
 The $o(1)$-term converges to $0$ as $w$ tends to infinity; thus, the growth rate of $w$ only 
 influences how fast this $o(1)$ error term vanishes, but (as far as we know) does not materially
 influence the success probability of PPSZ.
  We conclude:

\begin{theorem}[\cite{ppsz}]
\label{theorem-unique}
   If $F$ is a $k$-CNF formula with a unique satisfying assignment, then \begin{align*}
   \Pr[\ppsz(F) \textnormal{ succeeds}] \geq
   2^{ -n + s_k n - o(n)} \ . 
   \end{align*}
   Furthermore, $s_k  = \frac{\pi^2}{6k} + o(1/k)$.   
\end{theorem}

 \subsection{Previous improvements}
 
 \paragraph{Multiple satisfying assignments.}
 The analysis of \ppsza{} runs into trouble if $F$ contains multiple satisfying assignments.
 In their original paper~\cite{ppsz} they presented a workaround; unfortunately, this is quite
 technical and, for $k=3,4$, exponentially worse than the  bound of 
 Theorem~\ref{theorem-unique}. It was a breakthrough when  Hertli~\cite{Hertli2011} gave a 
 very general analysis of PPSZ showing that the ``Unique-SAT bound'' also holds in the presence
 of multiple satisfying assignments. Curiously, his proof takes
 the result ``$\Pr[\forced(x, \pi)] = s_k -o(1)$'' more or less as a black box and does not ask how such 
 a statement would have been obtained. Steinberger and myself~\cite{SchederSteinberger} later simplified
 Hertli's proof and obtained a certain unique-to-general lifting theorem that is also important for this work:
 
 \begin{theorem}[Unique-to-General lifting theorem~\cite{ppsz}]
    If the success probability of PPSZ is at least $2^{-n + s_k n + \epsilon n}$ on $k$-CNF formulas
    with a unique satisfying assignment, for some $\epsilon > 0$, 
    then it is at least $2^{-n + s_k n + \epsilon' n}$ on 
    $k$-CNF formulas with multiple solutions, too, for some (smaller) $\epsilon' > 0$.
 \end{theorem}
   
   \paragraph{Improved algorithms.}
Concerning the Unique-SAT case, Hertli~\cite{Hertli2014} designed an algorithm that is a variant of PPSZ
and achieves a success probability of $2^{ -n + s_3 n + \epsilon \, n}$ for $3$-CNF formulas with a unique 
satisfying assignment. 
Unfortunately,
the concrete value of $\epsilon$ is so tiny that Hertli did not even bother to determine it, and his approach is extremely specific to $3$-SAT, with no clear path how to 
generalize it to $k$-SAT. A result by Qin and 
Watanabe~\cite{QinWatanabe} strengthened Hertli's
result somewhat. More recently, \hkzz{}~\cite{HKZZ} published an 
algorithm called {\em biased-PPSZ}, a version of PPSZ
in which some guessed variables are decided by a biased coin; which variables and how biased, that
 depends on the structure 
of the underlying formula. In contrast to Hertli's, their improvement is ``visible'': for $3$-SAT, it improves the 
success probability from $1.3070319^{-n}$ in Theorem~\ref{theorem-unique} to
$\hkzzbasis^{-n}$. Also, it works for all $k$ (although the authors do not work out the exact magnitude of 
the improvement).

\paragraph{Lower bounds.}

Chen, Tang, Talebanfard, and myself~\cite{CSTT} have shown that there are instances on which 
PPSZ has exponentially small success probability. Just {\em how exponentially small} has been tightened by
Pudl\'ak, Talebanfard, and myself~\cite{PST}: we now 
know that PPSZ has success probability at most $2^{ - (1 - 2/ k - o(1/k)) \cdot n}$ on certain instances, provided our parameter $w$ is not too large; 
for {\em strong PPSZ}, i.e., PPSZ
using small-width resolution instead of $w$-implication, 
then the same bound holds, provided your width bound is really small,
like $c \cdot \sqrt{ \log \log n}$~\cite{SchederTalebanfard}.

\subsection{Our contribution}

We show that the success probability of PPSZ on $k$-CNF formulas
is exponentially larger than $2^{-n + s_k n}$. In particular,
\begin{theorem}[Improvement for all $k$]
\label{theorem-general}
   For every $k \geq 3$ there is $\epsilon_k > 0$ such that
   the success probability of PPSZ on satisfiable $k$-CNF formulas 
   is at least $2^{-n (1 - s_k - \epsilon_k)}$.
\end{theorem}

\paragraph{Comparison to Hansen et al.}
As already mentioned, a paper by \hkzz{}~\cite{HKZZ} introduces 
the algo\-rithm biased-PPSZ, which also exhibits an exponentially
improved running time. Which approach is better, ours or theirs? 
Numerically, neither paper bothers to analyze the asymptotic behavior
of the improvement $\epsilon_k$ as $k$ grows. Conceptually, one might 
argue that our result subsumes theirs, because they actually have to define
a new algorithm, while we simply give a better analysis of the old one. 
Methodologically, the two approaches are somewhat orthogonal: 
Hansen et al.~choose the Boolean values of 
the variables in a non-uniform 
way but leave the permutation of the variables
uniform; we change the permutation (more appropriately, we ``pretend'' 
to change it, because we don't change the algorithm) but choose 
the Boolean values uniformly. The fact that we do not change 
the algorithm seems 
like a limitation but actually gives us greater freedom: we can exploit information gleaned from the formula, 
even if that information is by itself NP-hard to compute.
I suspect that one can combine two approaches and get 
improved numbers for small $k$, like $k=3$; however, I fear 
that doing so would be extremely tedious and barely offer any additional 
insight. 

Personally, I think it would be more fruitful to focus on both approaches 
individually and explore how far each can be pushed because they, 
in the words of~\cite{HKZZ}, ``only scratch the surface''. 
Unfortunately, both currently suffer from the same shortcoming:
they do not improve 
the asymptotic $\frac{\pi^2}{6}$-factor in the behavior of the savings 
$s_k$ for large $k$. To be more precise, 
the improvement $\epsilon_k$ shrinks like $o\pfrac{1}{k}$ for 
both Hansen et al.~and this paper and thus becomes negligible 
compared to $s_k = \frac{\pi^2}{6} \cdot \frac{1}{k} + o\pfrac{1}{k}$.

\subsection{The case \texorpdfstring{$k=3$}{k=3}}

The case $k=3$ is the most visible and exhibits the fiercest competition. 
The full version of Hansen et al.~and 
the ECCC version of this result~\cite{Scheder2021ECCC}
invest considerable energy to hammer out a concrete numerical 
result how much they can improve over $s_3$. And although the $k=3$ part of \cite{Scheder2021ECCC} follows roughly the same approach as the 
general-$k$ case in this paper, it introduces several new concepts 
and methods that are not needed for Theorem~\ref{theorem-general}. 
Furthermore, it is highly technical, and the set of people 
interested in it is most likely a clear subset of those interested 
in the general-$k$ case. Finally, the analysis for $k=3$ in~\cite{Scheder2021ECCC} does not hit any natural wall, and therefore 
a simple tightening of inequalities and a better choice 
of constants and functions would already yield a better bound. 
We therefore decided not to include 
the $k=3$ part in this paper. 

\subsection{Organization of the paper}

We outline our general idea, analyzing PPSZ under some non-uniform
distribution on permutations, in Section~\ref{section-overview}.
Section~\ref{section-cct} introduces the notions of critical clause trees
and ``cuts'' in those trees. This is mainly a review of critical clause trees 
as defined in~\cite{ppsz}; however, since we will manipulate these
trees extensively, we introduce a more abstract and robust version, called
``labeled trees'' and cuts therein.
Section~\ref{section-general-k} contains the proof of 
Theorem~\ref{theorem-general}, our main result.
We strive for succinctness above all else and
took no effort to optimize the magnitude of our improvement.

\section{Brief overview of our method}
\label{section-overview}

\subsection{Working with a make-belief distribution on permutations}

Our starting point is to take a closer look at the application of Jensen's  inequality:
\begin{align*}
 \E_{\pi} \left[ 2^{-n + \forced(\pi)}\right]   & \geq  2^{ -n + \E_{\pi} [\forced(\pi)]} \ .
\end{align*}
This would be tight if $X := \forced(\pi)$ was the same for every permutation $\pi$. But maybe
certain permutations are ``better'' than others. The idea is to define a new distribution $D$
on permutations, different from the uniform distribution, under which ``good''
permutations have larger probability, thus $\E_{\pi \sim D} [X]  > \E_{\pi \sim U} [X]$.
Sadly, we have no control over the distribution of permutations: firstly, we promised not to change
the algorithm; secondly, and more importantly, defining $D$ will require  some information that is itself
NP-hard to come by. There is a little trick dealing with this. Generally speaking, 
if we want to bound the expression $\E_{Q} \left[2^X \right]$ from below but
the obvious bound from Jensen's inequality, $2^{\E_{Q}[X]}$, is not good enough for our purposes,
we can replace $Q$ by our favorite $P$ but have to pay a price. Formally:
\begin{align}
    \E_{Q} \left[2^X \right] & = \sum_{ \omega \in \Omega}  Q(\omega) 2^{X(\omega)}
     = \sum_{\omega \in \Omega} P(\omega) \cdot \frac{Q(\omega)}{P(\omega)} 2^{X(\omega)} \nonumber \\
   & = \E_{\omega \sim P}\left[ 2^{X(\omega) - \log_2 \frac{ P(\omega)}{Q(\omega)}}\right]\nonumber \\
   & \geq 2^{   \E_{P} [X] - \E_{\omega \sim P}\left[ \log_2 \frac{P(\omega)}{Q(\omega)}\right]} \nonumber \\
    & = 2^{ \E_{P}[X] - \KL(P || Q)} \ .
    \label{distorted-jensen}
\end{align}
The term $\KL(P || Q) := \sum_{\omega} P(\omega) \log_2 \pfrac{ P(\omega)}{Q(\omega)}$ is known as the 
Kullback-Leibler divergence from~$Q$ to $P$. If $Q$ and $P$ are continuous distributions 
(over $\Omega = [0,1]^n$, for example) 
with density functions $f_Q$ and $f_P$, then (\ref{distorted-jensen}) still holds, for 
$\KL(P || Q) := \int_{\Omega} f_P(\omega) \log_2 \pfrac{f_P(\omega)}{f_Q(\omega)}$.
This trick is not new: it plays a crucial rule in~\cite{SchederSteinberger}, and, if you look
close enough, also in Hertli~\cite{Hertli2011}; it appears, in simpler form, already in~\cite{ppsz}. 
However, in~\cite{SchederSteinberger, Hertli2011, ppsz},
the distribution $P$ is defined only to make ``liquid variables'' 
(variables $x$ for which $F|_{x = 0}$ and $F|_{x=1}$ are both satisfiable) 
come earlier in~$\pi$ and do not take the syntactic structure of $F$ into account: they define $P$ 
purely in terms of ${\rm sat}(F)$, the space of solutions, whereas our $P$ will depend heavily on the structure on $F$
as a $k$-CNF formula. Our work is the first to apply this method to improving PPSZ on formulas
with a unique satisfying assignment.
\\

\subsection{Good make-belief distributions for PPSZ---a rough sketch} 
\label{rough-sketch}

How can we apply this idea to the analysis of PPSZ? 
The challenge is to find a distribution $D$ under which $\E_{\pi \sim D}[\forced(\pi)]$ is 
larger than under the uniform distribution. Since we assume that $F$ has the unique satisfying
assignment $\alpha = (1,\dots,1)$, we can find, for every variable~$x$, a {\em critical clause}
of the form $(x \vee \bar{y} \vee \bar{z})$.\footnote{Our informal outline assumes $k=3$ to keep notation
simple.} Critical clauses play a crucial role in~\cite{ppsz}
and~\cite{HKZZ} as well. Imagine we change the distribution on permutations such that
$y$ tends to come a bit earlier than under the uniform distribution. 
It is easy to see that this can only decrease $\E[ \forced(y,\pi)]$ (which is bad)
and only increase $\E[\forced(a, \pi)]$ for all other variables. In particular,
it usually increases $\E[\forced(x,\pi)]$ (which is good). Now assume the literal $\bar{y}$ appears in a 
disproportionally large number of critical clauses. Then the beneficial effect
of pulling $y$ to the front of $\pi$ outweighs its adverse effect. Thus,
if there is a set $V' \subseteq V$ of variables with $|V'| = \Omega(n)$, and
each $y \in V'$  appears in a large number of critical clauses,
we can define a new distribution $P$ on permutations~$\pi$ under which 
variables $V'$ tend to come earlier than under the uniform distribution. 
Using~$P$ as our make-belief distribution in (\ref{distorted-jensen}), we obtain
a success probability that is exponentially larger than the baseline. 
This is what we call the ``highly irregular case'' below. 

The other extreme would the ``almost regular case'', namely that almost every variable~$x$ has exactly one critical clause and that almost every
 negative literal $\bar{y}$ appears in 
exactly two critical clauses. In this case, we find a matching $M$, i.e., a set of disjoint pairs 
of variables such that $\{y,z\} \in M$ implies that $(x \vee \bar{y} \vee \bar{z})$ is a 
critical clause of $F$, for some variable $x$. We then adapt the distribution on permutations
such that the location of $y$ and $z$ is positively correlated---either they both tend to come
late or they both tend to come early. This will have both (easily quantifiable) beneficial
effects and (more difficult to quantify) adverse effects. However, we will see that
the adverse effects can only be large if $\E_{\pi \sim U}[\forced(\pi)]$ is already
larger than $s_k n$ under the uniform distribution.

 \section{Critical clause trees, labeled trees, and cuts}
\label{section-cct}

 This section introduces the key notions of 
 critical clause trees and cuts, which were already defined 
 in~\cite{ppsz}. We introduce a more general notion that we call
 {\em labeled trees} and of 
 {\em cuts} therein. We will perform extensive ``tree surgery'', and this
 new terminology will allow us to state our results in a concise, rigorous,  
 and readable fashion.

   \subsection{The Critical Clause Tree}
\label{subsection-ccct}

	All notions and results in this subsection already appear in~\cite{ppsz}, 
	although we might phrase certain things a little different.  
   We assume that $\alpha = (1,\dots,1)$ is the unique satisfying assignment of our 
   input $k$-CNF formula $F$.
   That means that for every variable $x$, we can find a clause of the form 
   $(x \vee \bar{y}_2 \vee \dots \vee \bar{y}_k)$. This is called a {\em critical clause of $x$}. 
   For an integer $h \in \mathbb{N}$, a  
   {\em critical clause tree of $x$ of height $h$}  is a rooted tree $T_x$ of height 
   at most $h$ with a bunch of additional information:
   every node $u$ of $T_x$ has a {\em variable label} $\varlabel(u)$;
   if the depth of $u$ is less than $h$, it has a {\em clause label} $\clauselabel(u)$.
   The tree is constructed as follows:
   \begin{itemize}
       \item Initialize $T_x$ as consisting of a single root node, and set $\varlabel(\root) = x$.
       \item While some node $u$ of $T_x$ of depth less than $h$ does not have a clause label yet:
       \begin{enumerate}
       \item Let $\alpha_u$ be the assignment arising from $\alpha$ by setting to $0$ all the variables $y$ that appear
       as variable labels on the path from the root to $u$ (including both the 
       label of the root, which is $x$, and the label of $u$). Let $a :=\varlabel(u)$.
       In particular, $\alpha_u(a) = 0$.
       \label{cct-local-assignment}
       \item Pick a clause $C$ that is violated by $\alpha_u$
        (this exists since $\alpha$ is the unique satisfying assignment),
       and set $\clauselabel(u) := C$.
       \label{cct-choose-clause}
       \item For each negative literal $\bar{z} \in C$, create a new child of $u$ and give it variable label $z$.
       Note that $u$ has at most $k-1$ children.
       \label{point-add-child}
       \end{enumerate}
   \end{itemize}
   The  construction depends on the choice of $h$, so we should write 
   $T_x^{(h)}$ instead of $T_x$; however, the number $h$ will be
   fixed throughout, so we simply write $T_x$ for brevity.\footnote{Indeed, 
   strictly speaking, $T_x$  also 
   depends on the formula $F$ and the satisfying assignment $\alpha$, so 
   one should write $T_x^{(F, \alpha, h)}$. However, since $F$, $\alpha$, and $h$ always
   refer to the same formula, assignment, and integer parameter, there is no 
   gain in explicitly listing this dependency.
   }
   This tree is central to the analysis in~\cite{ppsz} and also~\cite{HKZZ} (but 
   curiously is completely absent in~\cite{Hertli2011} and~\cite{SchederSteinberger}). The {\em depth} of a node~$u$ in a tree $T$
   is the length of the path from the root to $u$; we abbreviate it as $d_T(u)$ or simply $d(u)$ if $T$ is understood.
   $T_x$ is a $(k-1)$-ary tree: every node has at most $k-1$ children,
   as Step~\ref{point-add-child} creates a child for every {\em negative} literal
   in $C$; since $\alpha = (1,\dots,1)$ satisfies $C$, there are at most $k-1$
   negative literals in $C$.       
   \begin{observation}[\cite{ppsz}]
      \label{obs-labels-antichain}
      Suppose $C = (y_1 \vee \dots \vee y_i \vee \bar{z}_1 \vee \dots \vee \bar{z}_j)$ 
         is the clause label of a node~$u$. Then
      \begin{enumerate}
         \item each variable among $y_1,\dots,y_i$ appears as the clause label of some
         (not necessarily proper) ancestor of $u$;
         \label{item-ancestor-labels}
         \item $u$ has $j$ children (note that $j =0$ might happen)
         whose variable labels are $z_1,\dots, z_j$.
        \label{item-children-labels}
         \item If $v$ is a proper descendant of $u$ in $T_x$ 
         then $\varlabel(u) \ne \varlabel(v)$.
         \label{item-antichain}
      \end{enumerate}
   \end{observation}
   \begin{proof}
   		Point 1 and 2 follow immediately from the construction process.
		To show Point 3, 
   	   let $a := \varlabel(u)$ and $b := \varlabel(v)$. We have to show that
	   $a \ne b$. Let
        $w$ be the parent of $v$, so~$w$ is a (not necessarily proper) 
       descendant of $u$, and let $C = \clauselabel(w)$. By 
       construction, in particular Point~\ref{point-add-child} above, 
       $C$ contains the negative literal $\bar{b}$.
       Since $\alpha_w$ violates $C$, by choice of $C$ in Point~\ref{cct-choose-clause},
       we conclude that $\alpha_w(b) = 1$.
       By definition of $\alpha_w$
       in Point~\ref{cct-local-assignment}, we observe that 
       $\alpha_w(a) = 0$ and therefore
       $a \ne b$.
   \end{proof}

   A {\em down-path} in a rooted tree is a sequence $u_0, \dots, u_t$ 
   where each $u_i$ is the parent of $u_{i+1}$. A {\em root-path} 
   is a path starting at the root. Note that a root path {\em always}
   is a down-path.      
    By Observation~\ref{obs-labels-antichain},
    no variable can appear twice or more 
   on a down-path in $T_x$.
   Critical clause trees are important because of the following lemma:
   \begin{lemma}[\cite{ppsz}]
   \label{lemma-forced}
       Suppose $w \geq (k-1)^{h+1}$, where $w$
       is the strength parameter of PPSZ as in Algorithm~\ref{algo-ppsz-fixed-pi}.
       Let $x \in V$, $\pi$ a permutation of $V$, and denote by
        $A$  the set of variables coming before~$x$ in $\pi$.
       If every path from the root of $T_x$ to a leaf at depth $h$ contains a node
       $u$ with $\varlabel(u) \in A$ 
       then $\forced(x, \pi) = 1$.
   \end{lemma}
   \begin{proof}
	  Call a node $u$ of $T_x$ {\em dead}
      if $\varlabel(u)$ comes strictly before $x$ in $\pi$; 
      call a node $v$ {\em reachable}
      if the path from the root to $v$ contains no dead nodes. Note that 
      the root itself is reachable. 
	  By the assumption in the lemma, no leaf at depth $h$ is reachable,
	  and therefore every reachable node $v$ has a clause label. Let $G$ 
	  be the set of clause labels of all reachable nodes. Note that
	  $|G| \leq 1 + (k-1) + (k-1)^2 + \cdots + (k-1)^{h} \leq \frac{(k-1)^{h+1} - 1}{k-2}
	  \leq w$. 
      
            When the loop of $\textsc{ppsz}(F,\pi,w)$ arrives $x$, the formula
      $F$ has already been reduced to $F' := F|_{A = 1}$, i.e., 
      all variables in 
      $A$ have been set to $1$. We have to show that 
      $F'$ implies $x=1$. Indeed, we show something stronger, namely 
      that $G' := G|_{A = 1}$ implies $x=1$.\\
      
      To show that $G'$ implies $(x=1)$, 
      let $\gamma'$ be a total assignment to the remaining variables $V \setminus A$
      with $\gamma'(x)=0$. We have to show that $\gamma'$ violates $G'$.
	  Equivalently, we have to show that every total assignment $\gamma$ on $V$
	  that sets $\gamma(x) = 0$ and $\gamma(a) = 1$ for all $a \in A$ violates $G$.
      For such a~$\gamma$, find a maximal root path
      $u_0, u_1, \dots, u_t$ of nodes in $T_x$ such that
      $\gamma(\varlabel(u_i)) = 0$ for all $u_i$ on that path. Such a path is non-empty
      because $\gamma(\varlabel(u_0)) = \gamma(\varlabel({\rm root})) = 
      \gamma(x) = 0$ by assumption. Can the last node $u_t$ be a leaf of 
      height $h$?
      Obviously not: by assumption, such a path from root to leaf would 
      contain node, let us say $u_i$, with $a := \varlabel(u_i)$ coming before~$x$, and thus $a \in A$ and $\gamma(a) = 1$. We conclude that $u_t$ is not a 
      leaf of height $h$ and therefore has a clause label $C$. We write 
      $C = (y_1 \vee \dots y_i \vee \bar{z}_1 \vee \dots \vee 
      \bar{z}_j)$ with $i+j = k$.
      Note that $i \geq 1$ since $\alpha = (1,\dots,1)$ satisfies $C$,
      but $j=0$ is possible. 
      By Point~\ref{item-ancestor-labels} of Observation~\ref{obs-labels-antichain},
      each of the variables $y_1,\dots, y_i$ is the variable label of 
      one of the $u_0,\dots,u_t$, and thus $\gamma(y_1) = \dots= \gamma(y_i) = 0$;
      by Point~\ref{item-children-labels} of Observation~\ref{obs-labels-antichain},
      each of the $z_1,\dots,z_j$ is the variable label of a child $v$ of $u_t$.
      By maximality of the path $u_0,\dots,u_t$, we have
      $\gamma(\varlabel(v)) = 1$, and thus 
      $\gamma(z_1) = \dots = \gamma(z_j) = 1$. In other words, $\gamma$
      violates $C$ and thus $G$.            
   \end{proof}
   
   From now on, we take $h = h(n)$ to be the largest integer such that $w \geq (k-1)^{h+1}$.
   Note that 
   $\lim_{n \rightarrow \infty} h(n) = \infty$ because $\lim_{n \rightarrow \infty} w(n) = \infty$.

   \subsection{The canonical critical clause tree}

      In Point~\ref{cct-choose-clause} of the construction process 
      for critical clause trees, we might have several violated
      clauses to choose from. For example,
      if $x$ has more than one critical clause, then 
	  there are several different critical clause trees 
	  $T_x$. In this section, we make 
things unique by introducing the concept of 
{\em canonical clauses} and {\em canonical clause trees}.      
     
     \begin{definition}[Canonical critical clause]
     Among all critical clauses of $x$, we choose one 
     and call it the {\em canonical critical clause of $x$}.
	\end{definition}
      
      This choice is arbitrary but considered fixed from 
      now on. Every variable has exactly one canonical critical clause.

      \begin{definition}[Canonical critical clause tree, CCCT]    
	  The {\em the canonical critical clause tree} of a variable $x$ of 	
      height $h$ is the critical clause tree $T_x$ that is 
      constructed as above, but additionally 
      adhering to the following tie-breaking rule
      in Point~\ref{cct-choose-clause}:
   	  if the canonical
      critical clause of $\varlabel(u)$ is violated by $\alpha_u$, pick it as 
      $\clauselabel(u)$;
      otherwise, pick the lexicographically first violated clause.
   \end{definition}
      
   In particular, if $T_x$ is the CCCT of $x$, then the clause label of its 
   root is the canonical critical clause of $x$.
    We distinguish between 
   canonical and non-canonical nodes:
   
   \begin{definition}[Canonical nodes]
   In the canonical critical 
   clause tree $T_x$, we call a node $v$ {\em canonical} if, 
   for every  node $u$ on the path from the root
   to $v$ (including $\root$ and $v$), $\varlabel(u)$
   has exactly one critical clause, and this clause is $\clauselabel(u)$.
   \end{definition}

   Suppose $u$ is a maximal non-canonical node in $T_x$ (i.e., either $u$
   is the root itself or the parent of $u$ is canonical), 
   and write $a = \varlabel(u)$. There are two reasons why $u$ can be non-canonical:
   (1) the variable $a$ has two or more critical clauses; (2) the clause label $C := \clauselabel(u)$ 
   is not a critical clause (it has at least two positive literals, like $(x \vee a \vee \bar{b})$, for example). In
   the latter case, $u$ has at most  $k-2$ children in $T_x$.    
   
   \paragraph{Remark.} For the improvement for general $k$ as presented 
   in this paper, the notions of canonical clause trees 
   and canonical nodes therein are 
   not strictly necessary, and we could prove some improvement 
   without introducing it. However, since I feel that this is 
   the correct way of thinking about critical clause trees and since 
   the $k=3$ case makes heavy use of these notions, I decided to include 
   it in this paper, for future reference.

   \subsection{Labeled trees and cuts}

   We will perform certain manipulations on critical clause trees. 
   The resulting trees will not be critical clause trees in any 
   meaningful sense any more (for example, they sometimes become infinite). 
	Therefore, we introduce the more general notion of a 
	{\em labeled tree}. We note that certain labeled trees 
	appear implicitly in~\cite{ppsz}; however, our work is the 
	first that explicitly defines and uses them.
   
   \begin{definition}[Labeled trees]
   \label{definition-labeled-tree}
      We assume some countably infinite set $L$ of labels 
      with $V \subseteq L$. 
      A {\em labeled tree} is a rooted tree $T$, possibly infinite, in which 
      \begin{enumerate}
      \item each node $u$ has a label $\varlabel(u) \in L$;
      \item 
      \label{def-labeled-tree-no-ancestor} no label 
      appears twice on a down-path; that is, if $u$ is a proper ancestor of $v$ in $T$, then
      $\varlabel(u) \ne \varlabel(v)$;
      \item each node is marked either as {\em canonical} or {\em non-canonical}; if $u$ is non-canonical
      then so are all of its children (and by induction all 
      of its descendants);
      \item each leaf of $T$ is marked as either a {\em safe leaf} or an {\em unsafe leaf}.
      \end{enumerate}
      A {\em safe path in $T$} is a path starting at the root that either ends at a safe leaf
      or is infinite. We write $\can(T_x)$ to denote the set of canonical 
      nodes in $T_x$. Furthermore, all labeled trees appearing in this paper 
      are {\em $(k-1)$-ary}: each node has at most $k-1$ children.
   \end{definition}
   
    A critical clause tree of height $h$  becomes a labeled tree by simply 
    marking leaves at depth~$h$ as  {\em safe leaves}
    and all other leaves (of smaller depth) as {\em unsafe leaves} and 
    removing all clause labels. 
    
    From now on, instead of viewing $\pi$ as a permutation of the variables $V$, we view it as a 
    {\em placement} $\pi: L \rightarrow [0,1]$ on the 
    countably infinite set of labels. If $\pi$ is sampled from some continuous
    distribution 
    (for example the uniform distribution), then $\pi$ is injective with probability~$1$; its restriction to $V$ defines a permutation,
    by sorting the variables from low-$\pi$ to high-$\pi$. 
        
    \begin{definition}[$\cut$ and $\cut_r$ and $\wcut_r$]
    \label{definition-cut}
    Let $T$ be a labeled tree, $x$ the 
    label of its root,  and $r \in [0,1]$. The event $\cut_r(T)$ is an event in the probability 
    space of all placements, defined
    as follows: mark a non-root vertex  $u$ as {\em dead} if $\pi(\varlabel(u)) < r$ and 
    {\em alive} otherwise; mark $\root$ as alive.
     Then $\cut_r(T)$ is the event that
     every safe path in $T$ contains at least one dead node. 
     $\cut(T)$ is the event $\cut_{\pi(x)}(T)$, i.e., all nodes $u$ with 
     $\pi(\varlabel(u)) < \pi(x)$ are marked dead.
     
     We define $\wcut_r(T)$ (``weak cut'') 
     exactly as $\cut_r(T)$, only that we 
     additionally mark the root as dead if $\pi(x) < r$.
     \end{definition}
     
     Note that there is no corresponding event $\wcut(T)$. Weak cuts 
    only make sense with respect to a particular $r \in [0,1]$. 
    \begin{figure}
     \begin{center}
      \includegraphics[width=\textwidth]{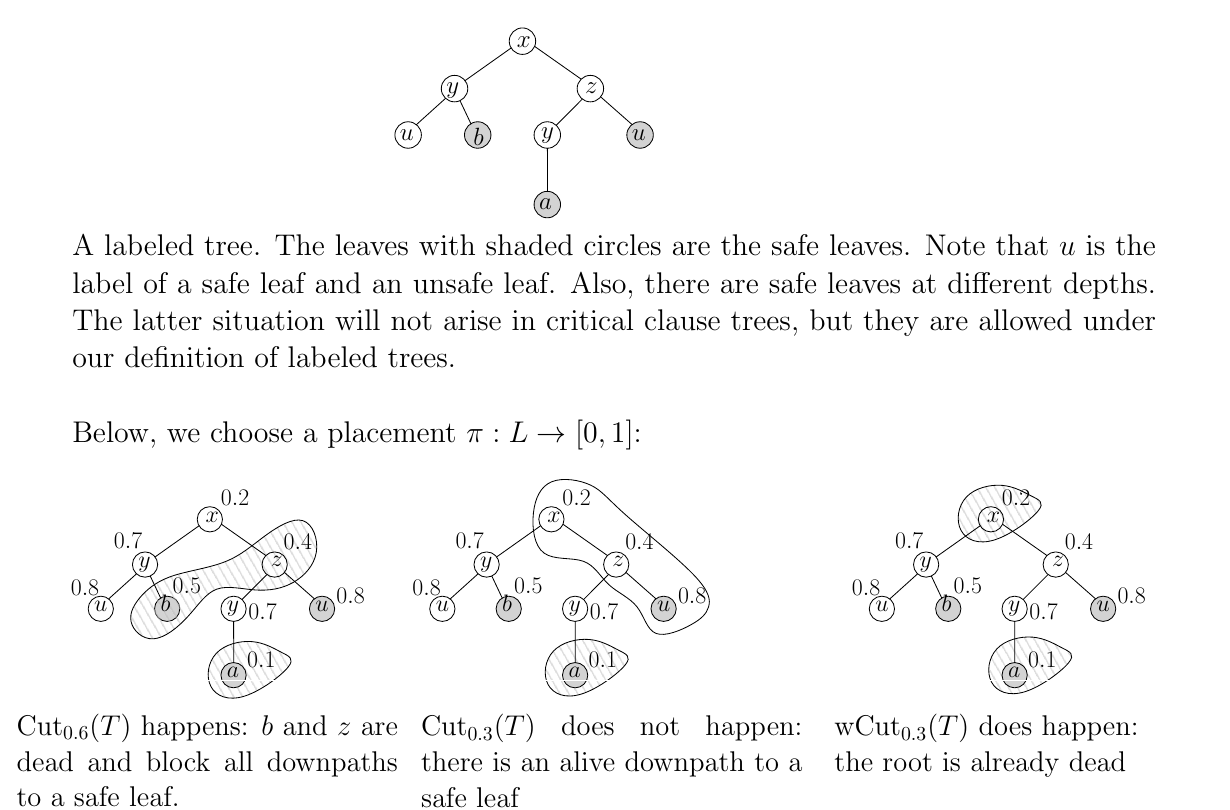}
     \end{center}
     \caption{A labeled tree, a placement, and a cut, a non-cut and a weak cut.}
    \end{figure}
     The following observation
     is simply Lemma~\ref{lemma-forced}, framed in the new terminology:
    
    \begin{observation}
       Suppose $w \geq (k-1)^{h+1}$.
        If $\cut(T_x)$ happens then $\forced(x, \pi) = 1$.
    \end{observation}
    To better understand the statement of the observation, 
    recall that the definition of the critical clause tree 
    $T_x$ depends on $h$, and the definition of $\forced$ depends 
    on $w$.
    The two definitions
    $\cut_r$ and $\wcut_r$ are intimately related: 
    
    \begin{observation}
    First, it holds that 
    \begin{align*}
    \wcut_r(T) = [ \pi(\root) < r \vee \cut_r(T)]
    \end{align*}
    Furthermore, if $T_1, \dots, T_{l}$  are the subtrees of 
    $T$ rooted at the children of the root,
    then 
    \begin{align*}
    \cut_r(T) = [ \wcut_r(T_1) \wedge \dots \wedge \wcut_r(T_{l})]
    \end{align*}
    \end{observation}
    
    A particularly important example of a labeled tree is $T_{k-1}^{\infty}$. 
    This is simply an infinite complete $(k-1)$-ary tree:
    every node has $k-1$ children, and there are no leaves. All nodes have distinct labels. If $k$ is understood, we simply
    write $T^\infty$.
    If $\pi$ is a uniformly random placement then 
    \begin{align*}
       Q^{(k)}_r := \Pr[\cut_r(T_{k-1}^\infty)] & = \left(\Pr[\wcut_r(T_{k-1}^\infty)]\right)^{k-1} \\
        P^{(k)}_r := \Pr[\wcut_r(T_{k-1}^\infty)]  & = r \vee \Pr[\cut_r(T_{k-1}^\infty)] \ , 
    \end{align*}
    where we define $a \vee b := a + b - ab$ for $a,b \in [0,1]$. 
    It follows from the two equations above that 
    $Q^{(k)}_r$ and $P^{(k)}_r$ satisfy the system of equations
    \begin{align*}
    Q & = P^{k-1}\\
    P & = r \vee Q \ . 
    \end{align*}
	By substitution, it follows that $Q^{(k)}_r$ and $P^{(k)}_r$ individually 
	satisfy the following equalities, respectively:
    \begin{align*}
    Q & = (r + (1-r) Q)^{k-1}\\
    P & = r \vee P^{k-1} \ . 
    \end{align*}
    Note that $P=1$ and $Q=1$ are always solutions of these equations, but 
    sometimes it is not the ``correct'' solution. 
    Indeed, it is a well-known result from the theory 
    of Galton-Watson branching processes that $Q^{(k)}_r$ and $P^{(k)}_r$ 
     are the {\em smallest} solutions of these equations.
    \begin{proposition}
    \label{prop-QR-convex}
        For $r \geq \frac{k-2}{k-1}$ it holds that $Q^{(k)}_r = P^{(k)}_r = 1$. On the interval $\left[ 0, \frac{k-2}{k-1}\right]$, 
        $P^{(k)}_r$ is convex and $r \leq P^{(k)}_r \leq \frac{k-1}{k-2} \cdot r$. 
        Also on that interval, $Q^{(k)}_{r} \leq \left( \frac{k-1}{k-2} \cdot r\right)^{k-1} \leq e\, r^{k-1}$.
    \end{proposition}    
        
    For $k=3$ we have explicit expression:
    \begin{align}
    \label{def-Q3}
       Q^{(3)}_r = \begin{cases}
       \pfrac{r}{1-r}^2 & \textnormal{ if $r < 1/2$} \\
       1 & \textnormal{ if $r \geq 1/2$} 
       \end{cases}
    \end{align}
    and
    \begin{align}
        \label{def-P3}
       P^{(3)}_r = \begin{cases}
       \frac{r}{1-r} & \textnormal{ if $r < 1/2$} \\
       1 & \textnormal{ if $r \geq 1/2$ }
       \end{cases}
    \end{align}
    Again, if $k$ is understood, we will simply write $Q_r$ and $P_r$.
    Recall the number $s_k$ as defined in (\ref{definition-sk})
     and observe that
    \begin{align}
    \label{def-sk}
        s_k = \Pr[\cut(T_\infty^{(k-1)})] \ .
    \end{align}

    \ppsza{} proved the following fact:
    \begin{lemma}[\cite{ppsz}]
    \label{lemma-cut-prob-usual}
       Let $T_x$ be a critical clause tree of height $h$. Then
       $\Pr[\cut_r(T_x)] \geq Q^{(k)}_r - \error(r, h)$ and 
       $\Pr[\cut(T_x) \geq s_k - \error(h)$.
       Here, 
		$\error(r,h)$ and $\error(h)$ are functions that 
       converge to $0$ as $h \rightarrow \infty$.
    \end{lemma}
    We will give a full and overly formal proof of this lemma.
    It will be instructive to go through the proof in full detail
    because it is a simplest application of our labeled-tree-machinery, on which
    we will rely more heavily later on. 
    Also, our proof is slightly different from the original proof in~\cite{ppsz}.
    \begin{proof}[Proof of Lemma~\ref{lemma-cut-prob-usual}]
       Let $T_x$ be the critical clause tree of height $h$ for variable $x$.
       We subject it to a sequence of transformation
       steps, each of which only reduces the probability
       $\Pr[\cut_r]$ and makes the tree look more and more like 
       $T_{\infty}^{(k-1)}$.\\
       
       \textbf{Step 1. Completing the tree.} Add nodes to $T_x$ to make
       it a complete $(k-1)$-ary tree of height $h$ (i.e., while some $v$
       of depth less than $h$ has fewer than $k-1$ children, create a new child).
       Give a fresh label from $L$ to each newly created node (by {\em fresh label}
       we mean a label in our infinite label space $L$ 
       that has never been used before) and mark all leaf as safe leaves (note
        that all leaves have height $h$ now). Call the resulting tree $T_0$.
       Every safe path in $T_x$ is still a safe path in $T_0$, and therefore
       $\Pr[\cut_r(T_0)] \leq \Pr[\cut_r(T_x)]$. Note that $T_0$ might not a 
       critical clause tree anymore (some non-leaves nodes do not have 
       clause labels, some variable labels are not variables at all but 
       fresh labels).\\
       
       \textbf{Step 2. Making labels distinct.} In a sequence of steps,
       we want to make sure that no label appears twice in the tree.
       To achieve this, we create a sequence $T_0, T_1, T_2, \dots$ as follows:
       
       \begin{proposition}
       \label{prop-fresh-label}
        Let $T$ be a labeled tree and $v$ a node therein. 
        Let $a'$ be a fresh label, i.e., 
        one that does not appear in $T$. Define $T_{v \rightarrow a'}$ to be 
        the same as $T$ but with $\varlabel(v) := a'$ for some 
        fresh label $a' \not \in L(T)$. Then 
        \begin{align*}
        \Pr[\cut_r(T_{v \rightarrow a'})] \leq \Pr[\cut_r(T_i)]
        \end{align*}
        holds.
       \end{proposition}

       This is how we construct the sequence $T_0, T_1, T_2, \dots$: 
       while some label (say $a$) appears more than once in $T_i$, 
       say at nodes $u$ and $v$, 
       take a fresh label $a'$ and set $T_{i+1} := T_{v \rightarrow a'}$.

       \begin{subproof}[Proof of Proposition~\ref{prop-fresh-label}]
          This is the heart of the proof. Fix  a partial assignment
          $\tau: L \setminus \{a, a'\} \rightarrow [0,1]$.
          We will show that $\Pr[\cut_r(T_{i+1}) \ | \ \tau] \leq 
          \Pr[\cut_r(T_i) \ | \ \tau]$. Here, the conditioning on $\tau$ is a shorthand
          for conditioning on the event $[\pi(l) = \tau(l) \ \forall l \in L \setminus
          \{a, a'\}]$.
          
          Note that both $\cut_r(T_i)$ and $\cut_r(T_{i+1})$ are monotone Boolean 
          functions in the atomic events $z_l := [\pi(l) < r]$ for $l \in L$.
          On the probability space conditioned on $\tau$, the event
          $\cut_r(T_{i+1})$ becomes a Boolean function $f(z_a, z_{a'})$ in the 
          Boolean variables $z_a := [\pi(a) < r]$ and $z_{a'} := [\pi(a') < r]$.
          Since $T_i$ can be obtained from $T_{i+1}$ by replacing $a'$ by $a$,
          we observe that conditioned on $\tau$, the event $\cut_r(T_i)$
          becomes the Boolean function $f(z_a, z_a)$. The proof of the claim
          now works by going through all possibilities what the monotone
          Boolean function $f$ could be.
          \begin{enumerate}
              \item If $f(z_a, z_{a'})$ is the constant $1$ function,
              then $\Pr[\cut_r(T_{i+1})  |  \tau] = 1
               = \Pr[\cut_r(T_i)  |  \tau]$; 
               similarly, if $f \equiv 0$ then both 
               probabilities are $0$.
               \item If $f(y,z) \equiv y$ then 
               $[\cut_r(T_{i+1}) |  \tau] = [\pi(a) < r] = 
               [\cut_r(T_i) |  \tau]$, so both events are the same.
               \item If $f(y,z) = z$ then 
               $[\cut_r(T_{i+1}) |  \tau] = [\pi(a') < r]$ and 
               $[\cut_r(T_i) | \tau] = [\pi(a) < r]$; the events are not 
               the same, but both have probability $r$.
               \item If $f(y,z) = y \wedge z$ then 
               {\small 
               \begin{align*}
               \Pr[\cut_r(T_{i+1}) | \tau] = \Pr[\pi(a) < r \wedge \pi(a') < r]
               = r^2 \leq r =  \Pr[\pi(a) < r]
               = \Pr[\cut_r(T_{i}) | \tau]
               \end{align*}
               }
               and the claimed inequality holds.
               \item If $f(y,z) = y \vee z$ then 
               {\small 
               \begin{align*}
               \Pr[\cut_r(T_{i+1}) | \tau] = \Pr[\pi(a) < r \vee \pi(a') < r]
               = 2r - r^2 > r =  \Pr[\pi(a) < r]
               = \Pr[\cut_r(T_{i}) | \tau]
               \end{align*}  
               }  
               and the claimed inequality does {\em not} hold. But here is the 
               thing: this cannot happen! Indeed, for 
               $[\cut_r(T_{i+1}) | \tau]$ to become 
               $[\pi(a) < r \vee \pi(a') < r]$, the two nodes $u$ and $v$
               would have to be ancestors of each other, which by 
               Observation~\ref{obs-labels-antichain} is impossible.
               To be more precise, suppose $[\cut_r(T_{i+1}) | \tau]$ is indeed
               $[\pi(a) < r \vee \pi(a') < r]$. Now set $\pi(a) = \pi(a') = 1$
               (and keep $\pi(l) = \tau(l)$ for all other labels $l$) 
               so the event does not happen. By definition of $\cut_r$, 
               this means that there is a 
               safe path $p$ in $T_{i+1}$ 
			    $\pi(\varlabel(w)) \geq r$ for all nodes $w \in p$.
			   If we change $\pi$ to $\pi'$ by setting $\pi'(a) = 0$ then
			   $[\pi'(a) < r]$ holds\footnote{Here
			   we silently assume $0 < r$. However, for the case $r=0$ it is 
			   clear that none of the atomic events $[\pi(l) < r]$ happen 
			   and thus $\cut_0(T_{i+1})$ and $\cut_0(T_i)$ are either both
			   the empty event or both the whole probability space and thus
			   happen with equal probability.}
			 and 
			   thus $\cut_r$ does happen; this means that
			   the variable label $a$ appears on the path $p$. By an analogous
			   argument, the label $a'$ appears on the path $p$, too.
			   In $T_i$, however, both those nodes have label $a$,
			   which contradicts 
			   Point~\ref{item-antichain} of 
			   Observation~\ref{obs-labels-antichain}.\footnote{Or rather it 
			   contradicts the definition of a labeled tree; the condition
			   that no ancestors-descendant pair have the same labels is 
			   baked into Definition~\ref{definition-labeled-tree},
			   and the reader should convince themselves that our transformations
			   within this proof keep this property.}			                                 
          \end{enumerate}
          Note that the claimed inequality holds in all possible cases.
          This completes the proof of Proposition~\ref{prop-fresh-label}.        
       \end{subproof}

       The sequence $T_0, T_1, T_2, \dots$ terminates because 
       each application of Proposition~\ref{prop-fresh-label} 
       introduces a new label to the tree, and the number of labels 
       will never exceed the number of nodes, which is 
	   $1 + (k-1) + (k-1)^2 + \cdots + (k-1)^{h}$.
	    Once it terminates, we obtain a tree $T_s$ in which all 
       labels are distinct and $\Pr[\cut_r(T_x)] \geq \Pr[\cut_r(T_s)]$.
       The event $\cut_r(T_s)$ can easily be phrased in terms of $T_{\infty}$: 
       for an integer $t$, let $\mathcal{C}_t$ be the event (in the probability
       space of placements $L \rightarrow [0,1]$) that all paths $p$ in $T_\infty$
       starting at the root and having length $t$ contain a non-root node $u$ 
       with $\pi(\varlabel(u)) < r$. Note that $T_s$
       is, in a sense, isomorphic to the first $h+1$ levels of~$T_\infty$,
       and therefore $\Pr[\cut_r(T_s)] = \Pr[\mathcal{C}_h]$.
       
       For given $\pi$, let $\mathbf{T}$
       be the subtree of $T_\infty$ obtained by deleting every
       non-root node $u$ with $\pi(\varlabel(u)) < r$, 
       together will all ancestors.
       Note that $\mathbf{T}$ is a random variable over our probability space
       of placements $\pi$. 
       Thus, $\mathcal{C}_t$ is the event that $\mathbf{T}$ has no root-path
       of length $t$, and $\cut_r(T_\infty)$ is the event 
       that $\mathbf{T}$ has no infinite root-path. Note that in the latter case,
       $\mathbf{T}$ is indeed finite and thus has a finite longest root-path, which 
       in turn means
       \begin{align*}
            \cut(T_\infty) = \bigcup_{t \in \N} \mathcal{C}_t \ .
       \end{align*}
       The $\mathcal{C}_t$ form an increasing sequence of events and thus,
       by the Monotone Convergence Theorem, it holds that 
       \begin{align}
             \Pr[\cut_r(T_\infty)] & = 
             \Pr\left[\bigcup_{t \in \N} \mathcal{C}_t  \right] 
              = \lim_{t \rightarrow \infty} \Pr[\mathcal{C}_t] \ .
              \label{eqn-monotone-convergence}
       \end{align}
       Define $\error(r,t) := \Pr[\cut_r(T_\infty)] - \Pr[\mathcal{C}_t]$.
       Then (\ref{eqn-monotone-convergence}) states that
       $\lim_{t \rightarrow \infty} \error(r,t) = 0$. We conclude that
       \begin{align*}
          \Pr[\cut(T_x)] & \geq \Pr[\cut(T_0)] \tag{by the transformation in Step 1} \\
          & \geq \Pr[\cut(T_s)] \tag{by the transformations in Step 2} \\
          & = \Pr[\mathcal{C}_h] \\
          & = \Pr[\cut_r(T_\infty)] - \error(r,h) \ .
       \end{align*}
       This concludes the proof of the lemma.       
    \end{proof}
    
    The important part of the proof, in particular with regard to what 
    comes below, is Proposition~\ref{prop-fresh-label}. In words, 
    the fact that multiple labels can only make things better. \ppsza{}~\cite{ppsz}
    use the FKG inequality to prove this. Our proof above does not use the 
    FKG inequality. Indeed, 
	from our proof technique one could extract 
    a stand-alone proof of the FKG inequality 
    (at least for the special case of
     monotone Boolean function and all variables being independent) that 
     does not use top-down induction but works by
     replacing variables by fresh copies, step by step, and showing that 
     each step satisfies the desired inequality.

\section{An Exponential Improvement}
\label{section-general-k}

\subsection{The highly irregular case}
\label{subsection-irregular-general-k}

\label{section-many-heavy}

In our outline in Section~\ref{rough-sketch} we mentioned that we will
address two cases separately: the ``highly irregular case'' and the 
``almost regular case''. The highly irregular case means that 
there is a set of $\Omega(n)$ variables each of which appears
as a negative literal in extraordinarily many critical clauses. This is the 
case that we treat in this section. Recall that every variable 
$x$ has one {\em canonical critical clause} that looks like 
$(x \vee \bar{y}_1 \vee \dots \bar{y}_{k-1})$.

\begin{definition}[Critical Clause Graph]
 The critical clause graph (CCG) is a directed graph on vertex set $V$ (the $n$ variables) 
 with $(k-1)n$ arcs defined as follows: for each $x \in V$, let 
 $(x \vee \bar{y}_1 \vee \dots \vee \bar{y}_{k-1})$ be its canonical critical clause and
  add the arcs $(x,y_1), \dots, (x,y_{k-1})$ to the graph.
\end{definition}

In the following figure, we show the beginning of a CNF formula with the canonical 
critical clauses underlined and part of its critical clause graph.

\begin{center}
\includegraphics[width=0.7\textwidth]{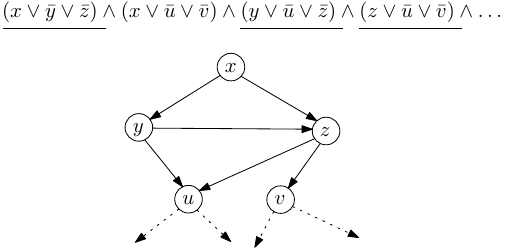}
\end{center}

Let us give a second example. 
If we take the simplest 3-CNF with a unique satisfying assignment, namely
\begin{align*}
 xyz \wedge xy\bar{z} \wedge x\bar{y}z \wedge \bar{x}yz \wedge 
\underline{ x\bar{y}\bar{z}} \wedge 
 \underline{\bar{x}y\bar{z}} \wedge 
\underline{\bar{x}\bar{y}z} 
\end{align*}
(where we replace $\vee$ by juxtaposition for compactness), we get the following 
critical clause graph:
\begin{center}
\includegraphics[width=0.25\textwidth]{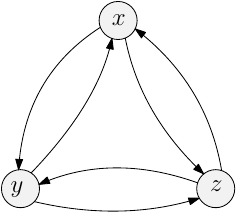}
\end{center}
We see that the CCG might have anti-parallel edges $(x,y), (y,x)$ but no 
self-loops. The latter follows since $(x \vee \bar{x} \vee \dots)$ always 
evaluates to $1$, and we assume no clause of $F$ has this form.
In the CCG, every vertex has out-degree $k-1$,
thus it has a total of $(k-1)n $ arcs, and the {\em average} in-degree
is $k-1$, too. 

\begin{definition}
For $x \in V$ let $\indeg(x)$ denote the in-degree of $x$ in the CCG. 
For a set $X \subseteq V$ we define 
\begin{align*}
    \indeg(X) := \sum_{x \in X} \indeg(x) \ . 
\end{align*}
An arc $(x,y)$ with $x, y$ both in $X$ 
counts towards $\indeg(X)$ as well. 
\end{definition}

For $l \in \N$, we call a variable $x$ {\em $l$-heavy}
if $\indeg(x) \geq l$. Let $\heavy(l)$ be the set of all 
$l$-heavy variables. If $l$ is understood from the context, we might 
simply write $\heavy$.

\begin{theorem}
\label{theorem-many-heavy}
    For every $k \in \N$, there is some $l = l(k) \in \N$ and 
    real number $\bonusheavy(k) > 0$
    (with $\bonusheavy$ standing for {\em heavy-bonus}) such that
    \begin{align*}
    \Pr[\textnormal{PPSZ succeeds}] \geq 2^{ -n + s_k n + \bonusheavy(k)\cdot\indeg(\heavy(l)) - o(n)} \ .
    \end{align*}
\end{theorem}
\begin{proof}
We define a new distribution $D$ on placements $\pi: V \rightarrow [0,1]$, in which 
$\pi(x)$ for $x \in \heavy$ 
is slightly biased towards smaller values. For this, 
we fix some continuous differentiable
$\gamma: [0,1] \rightarrow \mathbf{R}^+_0$ such that $\gamma(0) = \gamma(1) = 0$ and let
$\phi := \gamma'$ be its derivative. For sufficiently small $\epsilon > 0$, it holds that 
$1 + \epsilon \phi(r) \geq 0$ for all $r \in [0,1]$. This implies that 
$\int_0^r (1 + \epsilon \phi(x)) dx = r + \epsilon \gamma(r) =: \Phi(r)$
is monotonically increasing, $\Phi(0) = 0$ and $\Phi(1) = 1$. Thus, $\Phi$ it is the 
cumulative distribution function of a distribution on $[0,1]$.
\begin{definition}
 Let $D^{\gamma}_{\epsilon}$ be the distribution on $[0,1]$ with the cumulative 
 probability distribution function $r + \epsilon \gamma(r)$ and 
 density function $1 + \epsilon \phi(r)$.
\end{definition}

In this proof, we will state all intermediate results in terms of a general 
function $\gamma$ satisfying the above requirements. If we want to show 
{\em some} improvement for $k$-SAT, rather than concrete numerical results, 
the precise shape of $\gamma$ does not matter. We might just as well 
define $\gamma(r) = r(1-r)$ and run with it. In the following, 
we ask the reader to interpret $\gamma$ as the general function described above; 
however, since we will be juggling several constants, which depend on $k$ 
and our choice of $\gamma$, it will be advantageous sometimes to think of $\gamma$ 
as a very concrete function, not depending on anything, for example like 
$\gamma(r) = r (1-r)$. For this choice, we have $\phi(r) = 1 - 2r$, and 
$D^{\gamma}_{\epsilon}$ is a distribution on $[0,1]$ for all 
$0 \leq \epsilon \leq 1$.

Let $D$ be the distribution on placements $\pi: V \rightarrow [0,1]$ that samples $\pi(x) \in [0,1]$ 
uniformly for all $x \not \in \heavy$ and $\pi(x) \sim D^{\gamma}_{\epsilon}$ for all $x \in \heavy$.
For heavy $x$ it holds that $\Pr[ \pi(x) < r] = r +\epsilon \gamma(r) \geq r$. 
By using the indicator notation $[\textsc{Statement}]$ that 
gives $1$ if \textsc{Statement} holds and~$0$ otherwise, we can 
concisely write
\begin{align*}
   \Pr_{\pi \sim D} [\pi(x) < r] & = r + [x \in \heavy] \epsilon \gamma(r) \ .
\end{align*}

Loosely
speaking, heavy variables $x$ tend to come
{\em earlier} in $\pi \sim D$ than non-heavy ones. Consequently, for heavy $x$, we expect $\Pr[\forced(x,\pi)]$ to be 
{\em smaller} under $D$ than under the uniform distribution
$U_{[0,1]}$. The following lemma 
bounds the magnitude of this ``penalty'':

\begin{lemma}
   If $x \in \heavy$ then 
   $\Pr_{\pi \sim D} [\forced(x, \pi)]$ is at least 
   \begin{align*}
    \int_0^1 Q^{(k)}_r (1 + \epsilon \phi(r)) \, dr - o(1) = s_k - \epsilon 	\cdot \heavypenalty(k) - o(1) 
   \end{align*}
   for $\heavypenalty(k) := - \int_0^1 Q^{(k)}_r \phi(r) dr $,
    which depends only $k$ and on $\gamma$, but not on $k'$
   nor on $\epsilon$.
   The acronym $\heavypenalty$ stands for {\em heavy-penalty}.
\end{lemma}
Remark: the minus sign in the definition of $\heavypenalty$ might look surprising, 
but the integral $\int_0^1 Q^{(k)}_r \phi(r) dr$ is indeed 
at most $0$, which can be proved by basic calculus  or simply seen from the 
context (but this proof is not necessary for the correctness of the lemma).

\begin{subproof}
   Under $\pi \sim D$, the values $\pi(y)$ are independent
   for all variables $y$. Let $U$ denote the uniform 
   distribution on placements. Observe that we can define a coupling
   $(\pi_D, \pi_U)$ such that $\pi_D \sim D$ and $\pi_U \sim U$ and 
   $\pi_D(y) \leq \pi_U(y)$ hold always and for all variables $y$. 
   We could, for example, achieve this by sampling $u \in [0,1]$
   uniformly at random, let $t$ be the unique number in $[0,1]$
   for which $t + \epsilon \gamma(t) = u$, and set $\pi_U(y) = u$
   and $\pi_D(y) = t$. Then $\pi_U(y)$ is uniform over $[0,1]$
   and $\pi_D(y)$ is distributed according to $D_\epsilon^\gamma$.
   Furthermore, $t \leq u$ holds.
   
   Since $\cut_r(T_x)$ depends monotonically on the atomic
   events $[\pi(y) < r]$ and since $[\pi_U < r]$ implies $[\pi_D < r]$,
   we conclude that $\pi_U \in \cut_r(T_x)$ also implies 
   $\pi_D \in \cut_r(T_x)$; in other words,
   \begin{align}
      \Pr_D[\cut_r(T_x)] \geq \Pr_U[\cut_r(T_x)] \ .
      \label{ineq-D-better-than-U}
   \end{align}
   From the previous section, in particular Lemma~\ref{lemma-cut-prob-usual},
    we know that
   \begin{align}
   \Pr_U[\cut_r(T_x)] \geq Q_r^{(k)} - \error(r,h)\ . 
   \label{ineq-cut-prob-usual} 
   \end{align}
   The lemma now follows from setting $r = \pi(x)$ and 
   taking expectation over $\pi(x)$, and 
   using the fact that $\pi(x) \sim D$. That is,
   we compute 
   \begin{align*}
   \Pr_D[\cut(T_x)] & = \E_{r \sim D} \left[ \Pr_D [\cut(T_x)] \right] \\
   & = \int_0^1 \Pr_D [\cut_r(T_x)] (1 + \epsilon \phi(r)) dr 
   \tag{multiply by density and integrate}
   \\
   & \geq \int_0^1 \Pr_U [\cut_r(T_x)] (1 + \epsilon \phi(r)) dr  
   \tag{by (\ref{ineq-D-better-than-U})} \\
   & \geq \int_0^1 \left(Q_r^{(k)} - \error(r,h)\right)(1 + \epsilon \phi(r)) dr  
   \tag{by (\ref{ineq-cut-prob-usual})} \\
   & = \int_0^1 Q_r^{(k)} (1 + \epsilon \phi(r)) dr - 
   \int_0^1 \error(r,h) (1 + \epsilon \phi(r)) dr \ . 
   \end{align*}	
   Now since $\phi(r)$ is bounded and $\error(r,h)$ converges to $0$ as $h$ grows,
   for each fixed $r$, the second integral in the above expression also converges to $0$
   as $h$ grows, which explains
   the $o(1)$ term in the statement of the lemma. 
\end{subproof}

If $x \not \in \heavy$ but has a heavy out-neighbor, i.e., if there is an
arc $(x,y)$ for some $y \in \heavy$, we expect 
$\Pr[\forced(x,\pi)]$ to be {\em larger} under $D$. Also, if there are several 
heavy out-neighbors, we expect this ``bonus'' to be even larger. 
To formalize this intuition, we introduce some notation.

\begin{definition}
   For $x \in V$ and $Y \subseteq V$, we define $e(x,Y)$ to 
   be the number of arcs $(x,y)$ with 
   $y \in Y$. For sets $X,Y \subseteq V$ (not necessarily disjoint), 
   we define $e(X,Y)$ to be $\sum_{x \in X} e(x,Y)$.
\end{definition}

The next lemma formally states that the ``bonus'' for $x \not \in \heavy$ is proportional 
to its number of heavy out-neighbors:
\begin{lemma}
\label{lemma-bonus-heavy}
 If $x \not \in \heavy$ then
 $\Pr_{\pi \sim D} [\forced(x, \pi)]$ is at least
 \begin{align}
     s_k + \epsilon\, \heavychild(k)\, e(x, \heavy) - o(1) 
     \label{ineq-lemma-bonus-heavy}
 \end{align}
 where $e(x, \heavy)$ is the number of arcs $(x,y)$ with $y \in \heavy$ and 
 \begin{align*}
 \heavychild(k) :=  \int_0^1  \gamma(r)\left(P^{(k)}_r\right)^{k-2} \left(1 - Q^{(k)}_r\right) \, dr \ .
 \end{align*}  
   The acronym $\heavychild$ stands for {\em heavy-child bonus}.
  Note that $\heavychild(k)$ only depends on $\gamma$ and $k$.
	Furthermore, $\heavychild(k) > 0$ as long as $\gamma$ is positive 
	somewhere in $\left[ 0, \frac{k-2}{k-1}\right]$, which 
	our concrete choice $r ( 1-r)$ certainly is.
\end{lemma}

\begin{subproof}
	We first show the last part of the lemma, namely that 
	$\heavychild(k) > 0$. It is well-known  from the theory of Galton-Watson branching processes that $Q^{(k)}_r < 1$ for all $r < \frac{k-2}{k-1}$ 
	and thus
$\left(P^{(k)}_r\right)^{k-2} \left(1 - Q^{(k)}_r\right) > 0$ on the interval $\left[ 0, \frac{k-2}{k-1}\right]$. 
This means that $\heavychild(k)> 0$ for our choice 
$\gamma(r) = r(1-r)$ (and basically for any other legal 
choice for $\gamma$, provided that it is positive somewhere
on the interval $\left[ 0, \frac{k-2}{k-1}\right]$).\\

	We now prove (\ref{ineq-lemma-bonus-heavy}).
    Let $y_1,\dots,y_{k-1}$ be the labels of the children of the root of $T_x$ and $T_i$ be the 
    subtree of $T_x$ rooted at $y_i$.
    Similar to Step 2 in the proof of Lemma~\ref{lemma-cut-prob-usual}, 
    we can assume that all nodes of $T_x$ have distinct 
    variable labels. Formally, we need a more general version of 
    Proposition~\ref{prop-fresh-label}, one that deals with non-uniform 
    distributions:
       
       \begin{proposition}
       \label{prop-fresh-label-nonuniform-but-independent}
        Let $T$ be a labeled tree and $v$ a node therein. 
        Let $a'$ be a fresh label, i.e., 
        one that does not appear in $T$. Define $T_{v \rightarrow a'}$ to be 
        the same as $T$ but with $\varlabel(v) := a'$ for some 
        fresh label $a'$. Furthermore, let $D$ be a probability distribution
        on placements $\pi: L(T) \rightarrow [0,1]$ under which the variables 
        $\{\pi(l)\}_{l \in L}$ are independent. 
        Define a new distribution $D'$ by sampling $\pi(l) \sim D$ for all 
        labels $l \in L(T)$ and additionally sampling $\pi(a')$ from the same 
        distribution as~$\pi(a)$.
        Then 
        \begin{align*}
        \Pr_{D'}[\cut_r(T_{v \rightarrow a'})] \leq \Pr_D[\cut_r(T_i)]
        \end{align*}
        holds.
       \end{proposition}
       
    The proof of this proposition is almost the same as of 
    Proposition~\ref{prop-fresh-label}.
    For every variable $z \not \in \{x, y_1,\dots,y_{k-1}\}$, it holds that
    $\Pr_D[\pi(z) < r] \geq r$. This means we can assume 
    pessimistically\footnote{We will not go through a formal
    coupling argument from now on.} that $\pi(z)$ is uniform over $[0,1]$.
    Since all labels are distinct, we have $\Pr_D[\cut_r(T_x)] = \prod_{i=1}^{k-1} \Pr_D[\wcut_r(T_i)]$ and
    \begin{align*}
        \Pr_D[\wcut_r(T_{y_i})] & = \Pr_D[\pi(y_i) < r] \vee \Pr_D[\cut_r(T_{i})] \\
         & = (   r + [y_i \in \heavy]\epsilon\gamma(r) ) \vee \Pr_D[\cut_r(T_{i})] \\
         & \geq (   r + [y_i \in \heavy]\epsilon\gamma(r) ) \vee (Q_r  - o(1)) \\
         & = r + [y_i \in \heavy]\epsilon\gamma(r)  +
         \left( 1 - r - [y_i \in \heavy]\epsilon\gamma(r) \right) (Q_r - o(1)) \\
         & = r + (1-r) Q_r + [y_i \in \heavy]\epsilon\gamma(r) (1 - Q_r)  - o(1) \\
         & = P_r + [y_i \in \heavy]\epsilon\gamma(r) (1 - Q_r)  - o(1) \ . 
    \end{align*}
    Therefore,
    \begin{align*}
     \Pr_D[\cut_r(T_x)] & \geq \prod_{y: x \rightarrow y}
      \left( P_r + [y \in \heavy]\epsilon\gamma(r) (1 - Q_r) \right)
     - o(1) \\
     & \geq (P_r)^{k-1} + \sum_{y: x \rightarrow y}[y \in \heavy]\epsilon \gamma(r) 
     (P_r)^{k-2} (1 - Q_r)  - o(1) \\
     & = Q_r + e(x, \heavy)\epsilon\gamma(r) (P_r)^{k-2} (1 - Q_r)  - o(1) \ . 
    \end{align*}
    Since $\pi(x)$ is uniform over $[0,1]$, we get 
    $\Pr_{\pi \sim D}[\cut(T_x)]$ by integrating the above expression
    over $[0,1]$. This yields (\ref{ineq-lemma-bonus-heavy}) 
    and concludes the proof 
    of Lemma~\ref{lemma-bonus-heavy}.
\end{subproof}

We can prove Theorem~\ref{theorem-many-heavy} by summing over all variables: 
\begin{align}
   & \frac{ \sum_{x \in V} \Pr_D[\forced(x, \pi)] - s_k \, n}{\epsilon} \\
   & \geq - \heavypenalty(k) |\heavy| 
     + \heavychild(k) \sum_{x \not \in \heavy} e(x, \heavy) - o(n) \nonumber \\
     & \geq -\heavypenalty(k) |\heavy| + \heavychild(k) e(V \setminus \heavy, \heavy)
     - o(n) \ . \label{ineq-bonus-heavy-children} 
\end{align}
Next, we show that $e(V \setminus \heavy, \heavy)$ is large:
\begin{align*}
e( V \setminus \heavy, \heavy) & = e(V, \heavy) - e(\heavy, \heavy)\\
 & = \indeg(\heavy) - e(\heavy, \heavy)\\
&  \geq \indeg(\heavy) - e(\heavy, V) \\
& =  \indeg(\heavy) - k |\heavy| \ ,
\end{align*}
where the last equality follows because every variable, in particular every 
$x \in \heavy$, has $e(x,V) = {\rm outdeg}(x) = k-1$.
Combining the two, we get 
\begin{align*}
(\ref{ineq-bonus-heavy-children}) & = 
-\heavypenalty(k)|\heavy| + \heavychild(k)e(V \setminus \heavy, \heavy) -o(n)\\
     & \geq -\heavypenalty(k)|\heavy| + \heavychild(k)\left( \indeg(\heavy) - k |\heavy| \right) -o(n)\\
     & = \heavychild(k)\, \indeg(\heavy) - (\heavypenalty(k)+ k \heavychild(k)) |\heavy| - o(n)\\
     & 
\geq \left( \heavychild(k)- \frac{\heavypenalty(k)+ k \heavychild(k)}{l} \right)  \indeg(\heavy(l)) - o(n)\ , 
\end{align*}

where the last inequality follows from the fact that 
$\indeg(\heavy) \geq l |\heavy|$. We also write $\heavy(l)$ in the last 
line to emphasize that its definition involves the number $l$.
We choose $l = l(k)$ sufficiently large to make sure 
the expression in the parenthesis is some $c = c(k) > 0$. 
This shows that
$\E_D[\forced(\pi)] \geq s_k n - o(n) + \epsilon \, c(k) \, 
\indeg(\heavy(l(k))$.
Using (\ref{distorted-jensen}), we see that our ``gain'' in the exponent
of the success probability is at least\footnote{ignoring the $o(n)$ term for readability}
 \begin{align}
\log_2 \Pr[ \textnormal{PPSZ succeeds}] + n - s_k n \geq 
  \epsilon \, c(k) \, \indeg(\heavy) -  \KL(D || U) \ .
  \label{bonus-heavy-before-KL}
\end{align}
Since all values $\pi(x)$ are independent under both $D$ and $U$, the Kullback-Leibler divergence
becomes additive, and $\KL(D || U) = \KL (D^{\gamma}_{\epsilon} || U_{[0,1]} ) \cdot |\heavy|$, where 
$U_{[0,1]}$ is the uniform distribution on $[0,1]$.
\begin{proposition}
    Define $\Psi := \int_0^1 \phi^2(r)\, dr$. Then 
    $\KL(D^{\gamma}_{\epsilon} || U_{[0,1]}) \leq \log_2(e) \, \epsilon^2 \Psi$
    holds and furthermore
    $\KL(D|| U) \leq \log_2(e) \, \epsilon^2 \Psi |\heavy|$.
    \label{prop-bound-KL-manyheavy}
\end{proposition}
\begin{subproof}
     We abbreviate $t := \epsilon \phi (r)$. By definition of $\KL$ for continuous distributions, we have
     \begin{align*}
       \ln(2) \,  \KL(D^{\gamma}_{\epsilon} || U_{[0,1]}) 
       & = \int_0^1 (1 + t) \ln (1 + t) \, dr  \\
 &       \leq \int_0^1 (1 + t)  t\, dr \tag{since $\ln (1 + t) \leq t$}\\
       & = \int_0^1 t \, dr + \int_0^1 t^2\, dr \\
        & =  \int_0^1 \epsilon \phi(r) dr + \int_0^1 (\epsilon \phi(r))^2 \, dr 
        \tag{putting back $t = \epsilon\phi(r)$} \ . \\
     \end{align*}
     The first integral is $0$ because $\int_0^1 \phi(r) d r = \gamma(1) = 0$; 
     the second is $\epsilon^2 \Psi$ by definition of $\Psi$.
\end{subproof}
Combining Proposition~\ref{prop-bound-KL-manyheavy}  and (\ref{bonus-heavy-before-KL}) 
gives
\begin{align*}
\log_2 \Pr[ \textnormal{PPSZ succeeds}] + n - s_k n &  \geq 
  \epsilon \, c(k) \, \indeg(\heavy(l(k)) -  \log_2(e) \epsilon^2 \Psi |\heavy(l(k))| \\
  & \geq
  \left(\epsilon \, c(k) - \frac{\log_2(e) \epsilon^2 \Psi}{l(k)} \right) \indeg(\heavy(l(k)) 
\end{align*}

Since none of $c(k)$, $l(k)$, $k$, and $\Psi$ depends on $\epsilon$,
we can choose $\epsilon$ sufficiently small to ensure that the term 
in parentheses above some $\bonusheavy(k) > 0$ and thus 
$\log_2 \Pr[ \textnormal{PPSZ succeeds}] \geq - n + s_k n - o(n) + \bonusheavy(k) \cdot \indeg(\heavy(l(k))$, which 
proves Theorem~\ref{theorem-many-heavy}.
\end{proof}

Note that in theory we could hammer out how $\bonusheavy(k)$ depends on $k$. 
This would surely be tedious because we encountered several  integrals 
involving $Q^{(k)}_r$ and $P^{(k)}_r$, for which we don't have a closed form.
In its current form, the theorem gives us no explicit bounds on $\bonusheavy(k)$, 
other than that it is positive.

\subsection{Privileged variables---when \texorpdfstring{$\Pr[\forced(x,
    \pi)]$}{Pr[Forced(x,pi)]} is already larger}

There are some abnormal cases that will interfere with our analysis below. Luckily,
all those cases will imply that the variables involved already have a substantially
{\em higher} probability of being forced. 
\begin{definition}
\label{definition-privileged}
    A variable $x$ is called {\em privileged} if (1) $x$ has at least two critical clauses
    or it has a critical clause tree $T_x$ 
    such that (2) there is a variable $y$ that appears simultaneously at depth $1$ and $2$ 
    or (3) $T_x$ has fewer than $(k-1)^2$ nodes
    at depth $2$. Let $\setprivileged$ be the set of all privileged variables.
  \end{definition}
  \RestateInit{\restatelempriv}
\begin{restatable}{lemma}{lempriv}
    \label{thm:jl_rep}\RestateRemark{\restatelempriv}
\label{lemma-privileged}
   There is some value $\privlgd{} > 0$, depending only on $k$, such that
   $\Pr[\forced(x, \pi)] \geq s_k + \privlgd{} - o(1)$ for all privileged variables $x$,
   where $o(1)$ converges to $0$ as $w$ grows.
\end{restatable}

We prove the lemma in Appendix~\ref{appendix}.
The proof is somewhat technical, partially overlaps with proofs 
also found in Hansen et al.~\cite{HKZZ}, but is conceptually 
rather straightforward. 

\subsection{The almost regular case}
\label{section-general-k-almost-regular}
Theorem~\ref{theorem-many-heavy} already gives us an exponential improvement
over the old analysis of PPSZ provided that $\indeg(\heavy)$ is large (linear in $n$).
In this section, we will come up with a corresponding bound that works well if $\indeg(\heavy)$
is small. The final bound will then follow from a meet-in-the-middle argument.

\begin{lemma}
\label{lemma-large-backbone}
   There is a collection $G$ of canonical critical clauses 
   such that no two clauses in $G$ share a variable and $|G| \geq \frac{n - \indeg(\heavy(l))}{k l}$.
\end{lemma}

\begin{proof}
   Set $\mathcal{C}$ to be the set of all canonical critical clauses.
   We have $|\mathcal{C}| = n$ because there is exactly one critical clause for  
   each variable.
   Greedily pick a clause $C \in \mathcal{C}$, add it to $G$, and delete
   from $\mathcal{C}$ all clauses $C'$ that share a variable with $C$ (this obviously
   includes $C$ itself). Repeat this step for as long as possible.
   
   How many clauses does each step remove from $\mathcal{C}$? Let
   $x_1,\dots,x_k$ denote the variables of $C$. For sure we remove
   the canonical critical clauses of $x_1, \dots, x_k$. Additionally, we remove, for 
   each $1 \leq i \leq k$, all canonical critical clauses containing $\bar{x}_i$.
   This removes at most a total of $k + \sum_{x \in \var(C)} \indeg(x)$ clauses.
   Thus, the total number of canonical critical clauses removed in this process is at most 
   \begin{align*}
     &   \sum_{C \in G} \left( k + \sum_{x \in \var(C)} \indeg(x) \right) \\
     \leq &  |G|k + \sum_{C \in G} \sum_{\substack{x \in \var(C) \\ x \not \in \heavy(l)}} (l-1) 
     + \sum_{C \in G} \sum_{\substack{x \in \var(C) \\ x  \in \heavy(l)}} 
     \indeg(x)  \\
     \leq &  |G| k + |G| k (l-1) + \indeg(\heavy(l)) = |G| k l + \indeg(\heavy(l)) \ .
   \end{align*}
   The process ends $\mathcal{C}$ becomes empty.
   Since $|\mathcal{C}| = n$ in the beginning, it removes exactly $n$ clauses, and therefore
   $ |G| k l + \indeg(\heavy(l)) \geq n$.
   Solving for $|G|$ proves the lemma.
\end{proof}

We take a set $G$ of canonical critical clauses as guaranteed by the lemma. 
We form a collection $M'$ of disjoint pairs of variables by selecting, 
for each $C \in G$, two negative literals $\bar{y}, \bar{z} \in C$ 
 and adding $\{y,z\}$ to $M'$.\footnote{For $k=3$, the clause $C$ contains exactly
 two negative literals; for larger $k$, we select two literals arbitrarily.} 
 Call $x$ a {\em parent} of $\{y,z\}$ if 
the canonical critical clause of $x$ contains the literals $\bar{y}$ and $\bar{z}$. 
Note that this name makes sense since in the canonical critical clause tree $T_x$ of $x$,
the root has two children with labels $y$ and $z$, respectively.
Every $\{y,z\} \in M'$ has at least one parent.
We form a final collection $M \subseteq M'$ of pairs by removing each pair $\{y,z\}$ with parent $x$ from $M'$ 
if at least one of $x,y,z$ is in $\setprivileged$. Each privileged variable $x$ is ``responsible'' for the removal
of at most two elements from $M'$: one if the canonical clause of $x$ happens to be in $G$;
one if there is some $x'$ with $\{x, x'\} \in M'$. Therefore,
\begin{align}
\label{lower-bound-size-M}
|M| \geq \frac{ n - \indeg(\heavy)}{kk'} -  2\, |\setprivileged| \ . 
\end{align}
We denote the set of all parents $x$ of some $\{y,z\} \in M$ by $\parentM$.

\begin{theorem}
\label{theorem-large-M}
    For every $c > 0$ and $k \geq 3$ there is some $c' > 0$ such that
    if $|M| \geq c n$ then $\Pr[\textup{PPSZ succeeds}] \geq 2^{ -n + s_k n + c' n - o(1)}$.
\end{theorem}

From here, the proof of Theorem~\ref{theorem-general} is simple. Let $c_1$ be a small constant,
depending on $k$. If $|\indeg(\heavy)| \geq c_1 n $ then we can apply
Theorem~\ref{theorem-many-heavy}. 
If $|\setprivileged| \geq c_1 n$ we can apply 
Lemma~\ref{lemma-privileged}. Otherwise,  (\ref{lower-bound-size-M}) implies
that $|M| \geq c_2 n$ for some $c_2$ depending on $c_1$ and $k$, and $c_2 > 0$
if $c_1$ is small enough. We can now apply Theorem~\ref{theorem-large-M}
and are done. This proves Theorem~\ref{theorem-general}. It remains 
to prove Theorem~\ref{theorem-large-M}.\\

The remainder of this section contains the proof of Theorem~\ref{theorem-large-M}.
At first reading of what follows, it might even be helpful to think of  $\setprivileged$ as being empty.

\subsection{Using disjoint pairs to define a distribution}

We choose some $\rho \leq \frac{k-2}{k-1}$, to be determined later. 
Define $\gamma: [0,1] \rightarrow \mathbb{R}^+_0$ by
\begin{align}
\gamma(r) := \begin{cases}
r (\rho - r) & \textnormal{ if $r \leq \rho$}\\
0 & \textnormal{ if $r \geq \rho$.}
\end{cases}
\label{definition-gamma}
\end{align}
Let $\phi := \gamma'$ and extend this via $\phi(\rho) := - \rho$.
Observe that $\phi_{\min} := \min_{r \in [0,1]} \phi(r) = -\rho$.
  We fix  some $\epsilon > 0$. 
Let $D^{\gamma,\square}_{\epsilon}$ be the distribution on $[0,1] \times [0,1]$ whose density at $(r,s)$
is $1 + \epsilon \phi(r) \phi(s)$. 
This really is a  density, provided that $1 + \epsilon \phi(r) \phi(s) \geq 0$ for all $r$ and $s$.
Let $D$ be the distribution on placements that 
samples $(\pi(y), \pi(z)) \sim D^{\gamma,\square}_{\epsilon}$ for each $\{y,z\} \in M$
and
samples $\pi(x) \in [0,1]$ uniformly for each remaining variable. All samplings are done independently.
We define 
\begin{align}
\delta  = \delta(r) := \epsilon |\phi_{\min}| \gamma(r) = \epsilon \rho \gamma(r)
\label{def-delta}
\end{align}

\begin{lemma}
\label{lemma-benefit-damage}
   Let $r \in [0,1]$. Then 
   \begin{align*}
   \Pr_D[\cut(T_x) \ | \ \pi(x) = r] \geq Q_r - \damage(r) + \benefit(r) \cdot \mathbf{1}_{x \in \parentM} - o(1) \ , 
   \end{align*}
   where
   \begin{align}
       \damage(r) :=  \damageSpelledOut \label{def-damage} \\
       \benefit(r) :=  \benefitSpelledOut  \ . \label{def-benefit} 
   \end{align}
   and  the $o(1)$ converges to $0$ as $h$ grows.
\end{lemma}

\begin{proof}
The proof works by constructing an easy-to-analyze tree and distribution that serve as a pessimistic
estimate for $\cut_r(T_x)$. First, we make a simple but important observation
about~$M$ and the labels in $T_x$:
\begin{observation}
\label{observation-xy-not-in-M}
    Let $u$ be a node of depth $1$ in $T_x$ and $y = \varlabel(u)$. Then $\{x,y\} \not \in M$.
\end{observation}
\begin{proof}
   Suppose $\{x,y\} \in M$, for the sake of contradiction,
   and let $a$ be one of their parents.
   So $T_a$ contains two nodes $X, Y$ at depth $1$ with labels $x$ and $y$, respectively.
   What is the clause label $C_X$ of node $X$ in $T_a$? First, if it is $C_x$, the canonical critical 
   clause of $x$, then $X$ has a child with label $y$, since $\bar{y} \in C_x$.
   This means that $y$ appears in $T_a$ at level $1$ and $2$, so $a$ 
   is privileged.
   Second, if it is not $C_x$ but is {\em some} critical clause, it must be a critical clause for $x$
   or $a$, meaning that $x$ or $a$ has at least two critical clauses, so 
   at least one of them is privileged.
   Third, if $C_X$ is not a critical clause, then $C_X$ has at most $k-2$ negative literals and
   $X$ has at most $k-2$ children, meaning $T_a$ has fewer than $(k-1)^2$ nodes 
   at level $2$. In either case, at least one of $a$ and $x$ is privileged,
   meaning we would have eliminated $\{x,y\}$ from $M$. This is a contradiction.
\end{proof}
This observation has two important consequences:
\begin{observation}
\label{observation-not-parentM}
    Suppose $x \not \in \parentM$  and let $y_1,\dots,y_{k-1}$ be the labels of the depth-1-nodes in $T_x$.
    Then $\pi(x), \pi(y_1), \dots, \pi(y_{k-1})$ are independent and uniform under $D$.
\end{observation}
\begin{observation}
\label{observation-parentM}
    Suppose $x \in \parentM$ is a parent of $\{y,z\} \in M$. Let $y,z, v_1,\dots,v_{k-3}$ be the labels
    of the depth-1-nodes in $T_x$. Then
    $\pi(x), ( \pi(y), \pi(z)), \pi(v_1), \dots, \pi(v_{k-3})$ are independent under $D$, 
    the pair $( \pi(y), \pi(z))$ has distribution $D^{\gamma, \square}_{\epsilon}$, and the other 
    $k-2$ variables     are uniform over $[0,1]$.
\end{observation}
The upshot is that we completely understand the distribution of $\pi(l)$ for the labels on the 
level $0$ and $1$ of $T_x$. Starting from level $2$ downwards, the distribution can become complicated,
so we resort to a pessimistic estimate:
\begin{observation}
\label{obseration-r-delta}
   Let $v$ be a variable and let $\tau: V \setminus \{v\} \rightarrow [0,1]$ be a particular placement of all
   other variables. Let $r \in [0,1]$. Then
   $\Pr_D[\pi(v) < r \ | \ \tau] \geq r - \delta$,
   for $\delta$ as defined in (\ref{def-delta}).
\end{observation}
\begin{proof}
   If $v$ is not contained in any pair of $M$, then $\Pr_D[\pi(v) < r \ | \tau] = r$. If $\{v,w\} \in M$ then
   $\Pr_D[\pi(v) < r \ | \tau] = \Pr_{D^{\gamma,\square}_{\epsilon}}[\pi(v) < r \ | \ \pi(w) = \tau(w)]
   = r + \epsilon \gamma(r) \phi(\tau(w)) \geq r - \epsilon |\phi_{\min}| \gamma(r)$.
\end{proof}

\paragraph{The two pessimistic distributions $D_M$ and $D_{\bar{M}}$.}
Fix $r \in [0,1]$ and let $T^{\infty}$ be the complete infinite $(k-1)$-ary tree in which all labels
are distinct. Since all the labels are distinct, we take the liberty of writing $\pi(v)$ 
instead of $\pi(\varlabel(v))$ for a node $v$ in $T^{\infty}$. We specify two distributions 
$D_M$ and $D_{\bar{M}}$ on placements $L \rightarrow [0,1]$. 
First, we set $\Pr[ \pi(v) < r] := r - \delta$ under both $D_{\bar{M}}$ and $D_M$ for all
nodes $v$ of depth at least $2$ in $T^{\infty}$. Second, we 
let $y_1, \dots, y_{k-1}$ be the nodes 
of depth $1$ in $T^{\infty}$ and sample all $\pi(y_i) \in [0,1]$ uniformly and independently under $D_{\bar{M}}$.
Under $D_M$, we sample $\pi(y_3, \dots, y_{k-1})$ uniformly and independently but sample
$(\pi(y_1), \pi(y_2)) \sim D^{\gamma,\square}_{\epsilon}$.
This does not fully specify a distribution on placements $L \rightarrow [0,1]$ but 
it does specify the joint distribution of the events $[\pi(v) < r]$. Since we are only
interested in $\Pr[\cut_r(T_x)]$, this is enough. Note that we also 
do not need to specify a distribution for $\pi({\rm root})$.
\begin{lemma}
\label{obs-pessimistic-tree}
   If $x \not \in \parentM$ then $\Pr[\cut(T_x) \ | \ \pi(x) = r] \geq \Pr_{D_{\bar{M}}} [\cut_r(T^\infty)] - o(1)$.
   If $x  \in \parentM$ then $\Pr[\cut(T_x) \ | \ \pi(x) = r] \geq \Pr_{D_{M}} [\cut_r(T^\infty)] - o(1)$.
\end{lemma}
The $o(1)$-term comes from the fact that $T_x$ is a critical clause tree of height $h$, whereas
$T^{\infty}$ is an infinite tree.
\begin{proof}
  We mimic the proof of Lemma~\ref{lemma-cut-prob-usual} by assigning 
  each node at depth 2 or greater a fresh label.
  We need a third version of Proposition~\ref{prop-fresh-label}:
       \begin{proposition}
       \label{prop-fresh-label-correlated}
        Let $T$ be a labeled tree and $v$ a node therein. 
        Let $a'$ be a fresh label, i.e., 
        one that does not appear in $T$. Define $T_{v \rightarrow a'}$ to be 
        the same as $T$ but with $\varlabel(v) := a'$ for some 
        fresh label $a'$. 
        
        Let $D$ be a probability distribution
        on placements $\pi: L(T) \rightarrow [0,1]$ satisfying 
        the following property: for every label $l \in L(T)$ and 
        placement $\tau : L(T) \setminus \{l\} \rightarrow [0,1]$, 
        it holds that 
        $\Pr_D[\pi(l) < r \ | \ \tau] \geq r - \delta$.         
        Define a new distribution $D'$ by sampling 
        a placement $\pi: L(T) \rightarrow [0,1]$ from $D$ and then, 
        additionally, sampling $\pi(a')$ such that 
        $\Pr_{D'}[\pi(a') < r] = r - \delta$. Then 
        \begin{align*}
        \Pr_{D'}[\cut_r(T_{v \rightarrow a'})] \leq \Pr_{D}[\cut_r(T_i)]
        \end{align*}
        holds.
       \end{proposition}  
  
  		The proof of this proposition is almost the same as of 
		Proposition~\ref{prop-fresh-label}. We apply it repeatedly, 
		to each node of depth 2 or greater, noting that the distribution 
		$D$ always satisfies the requirement 
		$\Pr_D[\pi(l) < r \ | \ \tau] \geq r - \delta$, 
		by Observation~\ref{obs-pessimistic-tree}.         
		After this process, only 
	 the nodes at level 0 and 1 retain their old labels (and 
  thus their original joint distribution under $\pi$). We apply 
  Observations~\ref{observation-not-parentM} and~\ref{observation-parentM}
  to conclude that we have indeed arrived at distributions 
  $D_M$ and $D_{\bar{M}}$, respectively.
\end{proof}

\begin{proposition}
\label{prop-tree-M}
     With $\benefit(r) = \benefitSpelledOut$ as defined 
     in (\ref{def-benefit}),
     it holds that 
   $\Pr_{D_M} [\cut_r(T^{\infty})] \geq \Pr_{D_{\bar{M}}} [\cut_r(T^{\infty})] + \benefit(r)$. `
\end{proposition}
\begin{proof}
Let $T_1, \dots, T_{k-1}$ be the subtrees of $T^{\infty}$ rooted at the nodes of depth $1$.  
Let $\tau: L \setminus \{y_1, y_2\} \rightarrow [0,1]$
be a partial placement. Observe that the marginal distribution $D'$ 
of $\tau$ is the same under $D_M$ and $D_{\bar{M}}$.
We call $\tau$ {\em critical} if $\cut_r(T_1)$ and $\cut_r(T_2)$ do {\em not} happen
but $\wcut_r(T_3),\dots,\wcut_r(T_{k-1})$ do happen.
This can be determined by looking at $\tau$ alone.\\

\begin{claim}If $\tau$ is not critical then
$\Pr_{D_M}[\cut_r(T^{\infty}) \ | \ \tau] =
\Pr_{D_{\bar{M}}}[\cut_r(T^{\infty}) \ | \ \tau]$.
\end{claim}

\begin{proof}
There are several reasons why $\tau$ is not critical. First, if some 
of  $\wcut_r(T_3)$, $\dots$, $\wcut_r(T_{k-1})$ does not happen then
$\cut_r(T^\infty)$ will not happen, whatever the values of $\pi(y_1)$ and $\pi(y_2)$.
Second, suppose both of $\cut_r(T_1)$ and $\cut_r(T_2)$ happen.
Then $\wcut_r(T_i)$ happens for all $1 \leq i \leq k-1$ and $\cut_r(T^\infty)$ happens
for sure.
Third, suppose exactly one of $\cut_r(T_1)$ and $\cut_r(T_2)$ happens, say
$\cut_r(T_2)$. This means that $\cut_r(T^\infty)$ happens if and only if 
$\pi(y_1) < r$. This event has probability $r$ under both 
$D_{\bar{M}}$ and $D_M$. For $D_{\bar{M}}$ this is immediate from the definition;
for $D_M$ this follows from $D_\epsilon^{\gamma,\square}$ having uniform
marginals.
\end{proof}
If $\tau$ is critical then
\begin{align*}
    \Pr_{D_{\bar{M}}} [\cut_r(T^{\infty}) \ | \ \tau] & = \Pr_{D_{\bar{M}}}[\pi(y_1) < r \wedge \pi(y_2) < r] = r^2 \ , \\
    \Pr_{D_M}            [\cut_r(T^{\infty}) \ | \ \tau] & = \Pr_{D_M}[\pi(y_1) < r \wedge \pi(y_2) < r] = r^2 + \epsilon \gamma^2(r) \ .
\end{align*}
To bound the probability that $\tau$ is critical,
note that $\Pr_{D'} [\cut_r(T_i)] = Q_{r-\delta}$ since 
each of the atomic events $[\pi(l) < r]$ that are relevant for $\cut_r(T_i)$ 
happen with probability $r - \delta$. Furthermore, $Q_{r - \delta} \leq Q_r$. 
Next, $\Pr_{D'} [\wcut_r(T_i)] \geq P_{r - \delta}$ 
since each of the relevant atomic events happen with probability {\em at least}
$r - \delta$ (the event at the root of $T_i$, namely $[\pi(y_i) < r]$,
happens with probability exactly $r$, which is at least $r - \delta$). Therefore, 
\begin{align*}
   \Pr[\neg \cut_r(T_1)] \cdot \Pr[\neg \cut_r(T_2)] \cdot \prod_{i=3}^{k-1} \Pr[\wcut_r(T_i)]  
   \geq (1 - Q_r)^2 P_{r - \delta}^{k-3} \ , 
\end{align*}
and therefore
\begin{align*}
   \Pr_{D_{M}}[\cut_r(T^{\infty})] & =   \Pr_{D_{\bar{M}}}[\cut_r(T^{\infty})] + \epsilon \gamma^2(r) \Pr[\tau \textnormal{ is critical}] \\
   & \geq  \Pr_{D_{\bar{M}}}[\cut_r(T^{\infty})]  + \epsilon \gamma^2(r)  (1 - Q_r)^2 P_{r - \delta}^{k-3} \ .
\end{align*}
This completes the proof of Proposition~\ref{prop-tree-M}
\end{proof}

\begin{proposition}
\label{prop-tree-not-M}
Let $\damage(r) := \damageSpelledOut$ as defined in (\ref{def-damage}). Then 
it holds that 
$\Pr_{D_{\bar{M}}} [\cut_r(T^{\infty})]  \geq Q_r - \damage(r)$.
\end{proposition} 
\begin{proof}
   If $r \geq \frac{k-2}{k-1}$ then $\delta = 0$, $\damage = 0$, and $\Pr[\pi(v) < r] = r$
   for every node $v$ in $T^{\infty}$. Both sides of the inequality evaluate to $1$. \\
   
   Otherwise, let $T_1,\dots,T_{k-1}$ be the subtrees of $T^{\infty}$ rooted at the depth-1-nodes of $T^{\infty}$.
   Since $\Pr_{D_{\bar{M}}}[\pi(v) < r] = r - \delta$ for all nodes $v$ of $T^{\infty}$ of depth $2$ or greater,
   it holds that $\Pr_{D_{\bar{M}}}[ \cut_r(T_i)] = Q_{r- \delta}$ for $1 \leq i \leq k-1$. Since $Q_r$ is convex on 
   $\left[ 0, \frac{k-2}{k-1}\right]$, this is at least $Q_r - \delta Q'_r$. Next, 
   \begin{align*}
       \Pr_{D_{\bar{M}}} [\wcut_r(T_i)]  & = r \vee        \Pr_{D_{\bar{M}}} [\cut_r(T_i)] 
       \geq r \vee (Q_r - \delta Q'_r) \\
       & = r + (1-r) (Q_r - \delta Q'_r) = P_r - (1-r) \delta Q'_r \\
       & = P_r \left( 1 - \frac{(1-r) \delta Q'_r}{P_r} \right) \ ,
   \end{align*}
   and
   \begin{align*}
     \Pr_{D_{\bar{M}}} [\cut_r(T^{\infty})] & = \prod_{i=1}^{k-1}    \Pr_{D_{\bar{M}}} [\wcut_r(T_i)]  \\
     & \geq P_r^{k-1} \left( 1 - \frac{(1-r) \delta Q'_r}{P_r} \right)^{k-1} \\
     & \geq P_r^{k-1} \left(1 -  \frac{(k-1) (1-r) \delta Q'_r}{P_r} \right) \\
     & = Q_r - (k-1) (1-r) P_r^{k-2} \delta Q'_r \ .
   \end{align*}
   This completes the proof.
\end{proof}
From here on, we estimate for $x \not \in \parentM$:
\begin{align*}
    \Pr_D[\cut(T_x) \ | \ \pi(x) = r] & \geq \Pr_{D_{\bar{M}}} [\cut_r(T^{\infty})] - o(1) \tag{by Observation \ref{obs-pessimistic-tree}} \\
     & \geq Q_r - \damage(r) - o(1) \,  \tag{by Proposition~\ref{prop-tree-not-M}}  
\end{align*}
and for $x \in \parentM$:
\begin{align*}
    \Pr_D[\cut(T_x) \ | \ \pi(x) = r] & \geq \Pr_{D_{M}} [\cut_r(T^{\infty})] - o(1) \tag{by Observation \ref{obs-pessimistic-tree}} \\
     & \geq + \Pr_{D_{\bar{M}}} [\cut_r(T^{\infty})]  + \benefit(r) - o(1) \,  \tag{by Proposition~\ref{prop-tree-M}}  \\
     & \geq Q_r - \damage(r) +  \benefit(r) - o(1) \, \tag{ by Proposition~\ref{prop-tree-not-M}}  \ . 
\end{align*}
 This concludes the proof of Lemma~\ref{lemma-benefit-damage}.
\end{proof}

We obtain a lower bound on $\Pr_D[\cut(T_x)]$ by integrating the bound in Lemma~\ref{lemma-benefit-damage} over $r$:
\begin{align}
\label{prob-cut-damage-benefit}
    \Pr[\cut(T_x)] & \geq s_k - \damage + \benefit \cdot \mathbf{1}_{x \in \parentM} - o(1) \ , 
\end{align}
where $\damage = \int_0^ 1\damage(r) \,dr$ and $\benefit = \int_0^1 \benefit(r) \, dr$. To simplify the integration,
we will first give an upper bound on $\damage(r)$ and a lower bound on $\benefit(r)$. Recall our definition of 
$\gamma: [0,1] \rightarrow \mathbb{R}^+_0$
 in 
(\ref{definition-gamma}) as used in $D_\epsilon^{\gamma, \square}$:
\begin{align*}
\gamma(r) := \begin{cases}
r (\rho - r) & \textnormal{ if $r \leq \rho$}\\
0 & \textnormal{ if $r \geq \rho$.}
\end{cases}
\end{align*}

\begin{proposition}
\label{prop-damage-benefit-bounds}
   The following bounds hold:
    \begin{align}
         \damage & \leq O ( \epsilon \rho^{2k} ) \label{damage-upper-bound} \\ 
         \benefit & \geq \Omega( \epsilon \rho^{k+2} ) \ , \label{benefit-lower-bound}
    \end{align}
    where the $O$ hides factors depending solely on $k$ and terms of order $\rho^{a}$ for $a \geq 2\, k + 1$,
    and the $\Omega$ hides factors depending solely on $k$ and terms of order $\rho^b$ for $b \geq k  +2$.
\end{proposition}
\begin{proof}
     We remind the reader of the definitions of $\damage$ and $\benefit$ in (\ref{def-damage}) and (\ref{def-benefit}):
     \begin{align*}
         \damage(r) & = \damageSpelledOut \\
         \benefit(r) & = \benefitSpelledOut
     \end{align*}
     Both $\benefit(r)$ and $\damage(r)$ vanish for $r \geq \rho$. Thus, we can replace $\int_0^1$
     by $\int_0^{\rho}$. 
     We will first  bound $\benefit(r)$.
     On the interval $[0, \rho]$, $\gamma(r) = r (\rho - r)$, and $1 - Q_r \geq 1 - Q_{\rho}$,
     and $P_{r - \delta} \geq r - \delta = r (1 - \epsilon \rho (\rho - r))$. Therefore,
     \begin{align*}
        \benefit(r) & \geq \epsilon r^2 (\rho - r)^2 (1 - Q_{\rho})^2 r^{k-3} \left(1 - \epsilon \rho (\rho - r) \right)^{k-3} \\
         & \geq \frac{1}{2} \,  \epsilon r^{k-1} (\rho - r)^2  \ ,
     \end{align*}
     for sufficiently small $\epsilon$ and $\rho$ (smaller than a value depending solely on $k$).
     Integrating this over $r \in [0, \rho]$ shows (\ref{benefit-lower-bound}).\\
     
     Next, we bound $\damage(r)$ from above. 
     It holds that $r \leq P_r \leq \frac{k-1}{k-2} r$, where the first inequality follows immediately from
     $P_r = r \vee Q_r$, and the second follows from the fact that $P_r$ is convex on $\left[0, \frac{k-2}{k-1}\right]$,
     $P_0 = 0$, and $P_{\frac{k-2}{k-1}} = 1$.
     Third, we compute $Q'_r = (P_r^{k-1})' = (k-1) P_r^{k-2} P'_r \leq (k-1) \pfrac{k-1}{k-2}^{k-2} r^{k-2} P'_r$.
     To bound $P'_r$, observe again that $P'_r \leq P'_{\frac{k-2}{k-1}}$, where
     $P'_{\frac{k-2}{k-1}}$ is the {\em left derivative} since $P_r$ is not differentiable at $r = \frac{k-2}{k-1}$.
     To determine the left derivative, recall that $P = P_r $ satisfies the equation
     \begin{align*}
        P =   P^{k-1}  + (1 - P^{k-1}) r 
     \end{align*}
     and therefore
     \begin{align*}
          r = r(P) = \frac{P - P^{k-1}}{1 - P^{k-1}}  \ .
     \end{align*}
     For $P \rightarrow 1$, the derivative of  $P \mapsto r(P)$ converges 
     to $\frac{k-2}{2(k-1)}$ (compute the derivative of $r(P)$ and apply l'H\^{o}pital's rule twice; then 
     substitute $P=1$). Thus, for $r \rightarrow \frac{k-2}{k-1}$, the derivative $P'_r$ converges to the inverse
     thereof, to $\frac{2(k-1)}{k-2}$.
     Thus, $Q'_r \leq (k-1) \pfrac{k-1}{k-2}^{k-2} r^{k-2} \frac{2 (k-1)}{k-2}$.
     Putting things together and plugging in the definition of $\delta$ 
     in (\ref{def-delta}), we get
     \begin{align*}
        \damage(r) & = (k-1) (1 - r) P_r^{k-2} \delta Q'_r \\
        & \leq (k-1) (1-r) \pfrac{k-1}{k-2}^{k-2} r^{k-2} \epsilon \rho r (\rho -r) (k-1)  \pfrac{k-1}{k-2}^{k-2} 
        r^{k-2}
        \frac{2 (k-1)}{k-2} \\
        & = C_k \epsilon (1 - r)  r^{2k-3} \rho (\rho - r)
     \end{align*}
     for some constant $C_k$ depending only on $k$. Integrating over $r \in [0, \rho]$ 
     yields (\ref{damage-upper-bound}).
\end{proof}

Combining (\ref{prob-cut-damage-benefit}) with the bounds (\ref{damage-upper-bound}) 
and (\ref{benefit-lower-bound}) and summing over all $x \in V$, we obtain

\begin{align}
\label{cut-prob-sum-vars}
 \sum_{x \in V} \Pr_D[\cut(T_x)] & \geq s_k n - O( \epsilon \rho^{2k}) n + \Omega(\epsilon \rho^{k+2}) |M| 
- o(n) \ .
\end{align}

Finally, to bound the success probability of PPSZ using the distribution $D$, we need to bound
$\KL(D||U)$ from above.
By additivity of $\KL$, we see that $\KL(D||U) = |M| \cdot \KL( D^{\gamma,\square}_{\epsilon} || U^{\square})$,
where $U^{\square}$ denotes the uniform distribution on $[0,1] \times [0,1]$.
The density of $D^{\gamma,\square}_{\epsilon}$ at $r,s$ is $1 + \epsilon \phi(r) \phi(s)$, and therefore
\begin{align*}
     KL( D^{\gamma,\square}_{\epsilon} || U^2) & = 
     \frac{1}{\ln(2)} \, \int_{[0,1]^2} (1 + \epsilon \phi(r) \phi(s)) \ln ( 1 + \epsilon \phi(r) \phi(s)) \, ds \, dr \\
     & \leq 
     \frac{1}{\ln(2)} \, \int_{[0,1]^2} (1 + \epsilon \phi(r) \phi(s))  \epsilon \phi(r) \phi(s)) \, ds \, dr \\     
     & = 
     \frac{1}{\ln(2)} \, \int_{[0,1]^2} \epsilon^2 \phi^2(r) \phi^2(s) \, ds \, dr \\
     & = \frac{\epsilon^2}{\ln(2)} \left( \int_{0}^1 \phi^2(r) \, dr \right)^2 \\
     & = \frac{\epsilon^2 \rho^3}{3\, \ln(2)} \ .
\end{align*}
Thus, using (\ref{distorted-jensen}), we conclude that the success probability 
of PPSZ is $2^{ -n + s_k n + \text{gain} }$ where
\begin{align*}
\text{gain} & \geq \Omega( \epsilon \rho^{k+2}) |M| - O( \epsilon^2 \rho^3 ) |M| 
- O(\epsilon \rho^{2k}) n \ . 
\end{align*}
Choosing $\epsilon = \rho^{k-3}$, this is 
$\Omega(\rho^{2k - 1}) |M| - O( \rho^{3k - 3}) n$.
Note that $2k - 1 \leq 3k -3$ for $k \geq 3$. 
Thus, if $|M| \geq c n$, we can choose a sufficiently small $\rho$, depending on $k$ and $c$,
ensuring that $\text{gain} \geq c' n$, for some constant $c'$ depending
on $c$ and $k$.
This concludes the proof of Theorem~\ref{theorem-large-M}.

\section*{Acknowledgments}

I am grateful to Robin Moser and Timon Hertli, with whom I intensively discussed the PPSZ algorithm, and some ideas 
used in this paper came from our discussions. I am grateful to Navid Talebanfard for fruitful discussions and helpful comments on this work.
Finally, my thanks go to the anonymous reviewers, both of the conference 
version and of this version, which greatly helped me to improve 
the paper.

\printbibliography

\appendix

\section{Proofs from Section~\ref{section-cct} and~\ref{section-general-k}}
\label{appendix}

\RestateGo{\restatelempriv}
\lempriv*
\begin{proof}
  We show that there are constants $c_1, c_2, c_3 > 0$, depending only on $k$,
  such that $\Pr[\forced(x, \pi)] \geq s_k + c_i - o(1)$ whenever $x$ is privileged
  due to reason $(i)$ in Definition~\ref{definition-privileged}.
  
  The case of privileged variables of type (1), i.e., those having at least two critical clauses,
  has already been addressed in the full versions of~\cite{HKZZ}. However, for the sake
  of completeness we will also discuss this case.
  We will introduce some operations on labeled trees $T$ that never increase $\Pr[\cut_r(T)]$.
  As a most simple example, suppose $u$ is a node in $T$ and not a safe leaf; form 
  $T'$ by adding a new child $v$ to $u$. Then $\Pr[\cut_r(T)] \geq \Pr[\cut_r(T')]$, regardless
   of the label of $v$. This follows immediately from the definition of $\cut_r$.
  \begin{figure}
       \includegraphics[width=\textwidth]{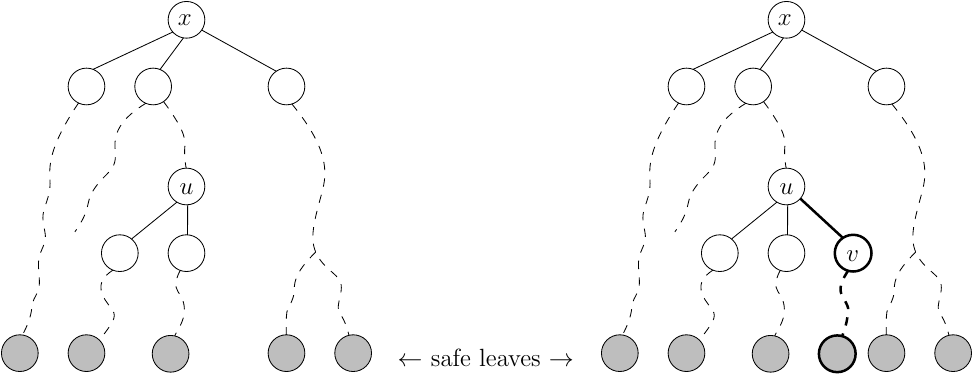}
       \caption{Attaching additional descendants to $u$ will not increase $\Pr[\cut_r(T)]$.}
  \end{figure}
  
  This operation allows us to reduce case (3) to case (2). Indeed, suppose $T_x$ has fewer 
  than $(k-1)^2$ nodes at depth $2$. Let $Y_1,\dots,Y_{k-1}$ be the children of the root of $T_x$
  and let $y_1,\dots, y_{k-1}$ be their respective labels. By assumption, some child $Y_i$ has at most 
  $k-2$ children. Create a new node $Z$, attach it as an additional child to $Y_i$, and give it label $y_1$.
  The resulting tree $T'$ is a labeled tree, every node has at most $k-1$ children, and 
  label $y_1$ occurs at depths $1$ (at $Y_1$) and at depth $2$ (at $Z$), so $T'$ is of type (2). 
  
  Next, we (almost) reduce case (1) to case (2). Suppose $x$ has two critical clauses,
  $C = (x \vee \bar{y}_1, \dots, \bar{y}_{k-1})$ and 
  $D = (x \vee \bar{z}_1, \dots, \bar{z}_{k-1})$.  
  Note that $k \leq | \{y_1, \dots, y_{k-1}, z_1, \dots, z_{k-1} \} | \leq 2(k-1)$.
  Suppose for the moment that it is less than $2(k-1)$, i.e, some variable $y_i$ also appears in $D$.
  Without loss of generality, $y_1 = z_1$. Also, the two clauses are distinct,
  so let us assume that $y_{k-1}$ does not appear in $D$.
  We will construct a (non-canonical) critical clause tree $T'_x$
  for $x$ that is of type (2). Use $C$ as clause label for the root. Note that this creates 
  $k-1$ nodes $Y_1, \dots, Y_{k-1}$ at depth~$1$ with labels $y_1,\dots,y_{k-1}$. The assignment label
  of $Y_{k-1}$ is $\alpha[y_{k-1} \mapsto 0]$, which violates $D$; here we use the fact that 
  $D$ does not contain $y_{k-1}$. Thus, we can use $D$ as clause label of node $Y_{k-1}$,
  which in turn creates $k-1$ nodes at depth 2, one of which has label $z_{1}$. Now recall
  that $y_1 = z_1$ by assumption, so this label occurs once at depth $1$ and somewhere at depth $2$,
  and $T_x$ is of type (2).
  
  Summarizing, we are left with privileged variables of type (2) and those with two critical
  clauses $C,D$ that share no variable besides $x$. Let us deal with type (2) first.
  We start with a proposition stating that assigning ``fresh labels'' to a node of $T$ cannot 
  increase $\Pr[\cut_r(T)]$. This can be seen as an alternative proof 
  of Lemma~7 in~\cite{ppsz} that bypasses the FKG inequality for monotone Boolean functions 
  (in fact, implicitly reproves it). 
  \begin{proposition}
      Let $T$ be a labeled tree and suppose the values $\{\pi(l)\}_{l \in L}$ are independent. Let $u$ be a node in $T$ 
      with label $z$. Form a new tree $T'$ by assigning $u$ a fresh label $z'$ and making
      $\pi(z')$ follow the same distribution as $\pi(z)$, but independent of everything else. Then 
      $\Pr[\cut(T)] \geq \Pr[\cut(T')]$
  \end{proposition}   
  \begin{proof}
      We show that $\Pr[\cut_r(T)] \geq \Pr[\cut_r(T')]$ for all $r \in [0,1]$. Let $\tau: L \setminus \{z, z'\} \rightarrow [0,1]$
      be a placement of all labels except $z$ and $z'$. In fact, we claim that
      $\Pr[\cut_r(T) \ | \ \tau] \geq \Pr[\cut_r(T')\ | \ \tau]$ holds for all $\tau$. Under the partial placement $\tau$,
      the event $\cut_r(T')$ becomes some monotone Boolean function $f(b,b')$ in the Boolean variables
      $b := [\pi(z) < r]$ and $b' := [\pi(z) < r']$, and $\cut_r(T)$ becomes $f(b,b)$. This holds since $T$ can be obtained
      from $T'$ by merging the labels $z$ and~$z'$. Now under our distribution on placements,
      $\pi(z)$ and $\pi(z')$ follow the same distribution and therefore
      $\Pr[b=1] = \Pr[b'=1]$. This means that $\Pr[f(b,b') = 1] = \Pr[f(b,b)=1]$ unless~$f$ depends on both variables;
      the only monotone functions depending on both $b$ and $b'$ are $b \wedge b'$ and $b \vee b'$.
      If $f(b,b') = b \wedge b'$ then $\Pr[ \cut_r(T') \ | \ \tau] = \Pr[ b \wedge b'] \leq \Pr[b] = \Pr[ \cut_r(T)]$ and our claim
      holds. Finally, $f(b,b') = b \vee b'$ cannot hold: the set of nodes in $T'$ with label $z$ or $z'$ form 
      an antichain, by Point~\ref{def-labeled-tree-no-ancestor}  of Definition~\ref{definition-labeled-tree}.
  \end{proof}  
  
   Now let $T$ be a labeled tree to which  (2) applies, i.e., some variable $y$ appears at depth $1$ and $2$
   in $T_x$. Let $Y_1, Y_2$ be those two nodes with label $y$. 
   We apply the proposition to all nodes except $Y_1$ and $Y_2$. 
   Second, we add new children to nodes of depth less than $h$ 
   until all such nodes have exactly $k-1$ children, and all the $(k-1)^h$ nodes at depth $h$ are safe leaves.
   Call this tree $T$. From the proposition and the discussion above, it follows that 
   $\Pr[\cut_r(T_x)] \geq \Pr[\cut_r(T)]$. 
   In a last step, give a fresh label $y_1$ to $Y_1$ and $y_2$ to $Y_2$, and call the resulting tree $T'$.
   In $T'$, all nodes have distinct labels, and we know what $\Pr[\cut_r(T')]$ is: 
   it is $Q_r^{(k-1)} - o(1)$.\footnote{We will write $Q_r$ instead of $Q_r^{(k-1)}$ from now on since 
   $k$ is understood.} 
   We also know that $\Pr[\cut_r(T)] \geq \Pr[\cut_r(T')]$ by the above proposition. Now, however,
    we have to take a closer look at the proof of the proposition since we want to show that
    $\Pr[\cut_r(T)]$ is significantly larger than $\Pr[\cut_r(T')]$. 
    \begin{center}
        \includegraphics[width=0.4\textwidth]{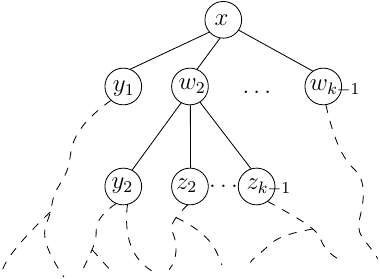}\\
        The tree $T'$: all labels are distinct.
    \end{center}
     In $T'$, we denote the node with label $w_2$ by $W_2$; that with label $z_2$ by $Z_2$, and so on.
     Since $T'$ and $T$ has the same node set, we use this notation for the nodes in $T$, too.    
     Furthermore, for a node $u$, we denote the subtree of $T$ (or $T'$) rooted at node $u$ by $T_u$ (or $T'_u$).
    As in the proof of the proposition, we fix some partial placement 
    $\tau: L \setminus \{y_1, y_2, y\} \rightarrow [0,1]$. As we have seen in the proof,
    $\Pr[\cut_r(T) \ | \ \tau] \geq \Pr[\cut_r(T') \ | \tau]$ holds for every such $\tau$.
    Call $\tau$ {\em good} if the following holds: 
    (1) $\pi(w_2) \geq r$; (2) $\pi(z_2), \dots, \pi(z_{k-1}), \pi(w_3), \dots, \pi(w_{k-1}) < r$;
    (3) $\neg \cut_r(T_{Y_1})$ and $\neg \cut_r(T_{Y_2})$.
    The events described in (1--3) are independent; those in (1) and (2) happen with probability 
    exactly $(1-r) r^{ 2k-5}$. Those in (3) happen with probability at least $(1 - Q_r)^2$.
     
    Under a good $\tau$, the $\cut_r(T')$ becomes $[\pi(y_1) < r \wedge \pi(y_2) < r]$ and has 
    probability $r^2$, and $\cut_r(T)$ becomes $[\pi(y) < r]$, which has probability $r$. Therefore,
    \begin{align*}
       \Pr[\cut_r(T)] - \Pr[\cut_r(T')] & \geq \Pr[\tau \textnormal{ is good}] \cdot (r - r^2) \\
       & \geq (1-r) r^{2k - 5} (1 - Q_r)^2 (r - r^2) 
       = (1-r)^2 r^{2k - 4} (1 - Q_r)^2 \ .
    \end{align*}
    
    Putting everything together and integrating over $r$, we conclude that
    \begin{align*}
       \Pr[\cut(T_x)] \geq s_k - o(1) + \int_0^1 (1-r)^2 r^{2k - 4} (1 - Q_r)^2 \, dr \  .
    \end{align*}
    It is clear that the integral is some positive constant depending solely on $k$.\\
    
    We are left with the case that $x$ has two critical clauses   $C = (x \vee \bar{y}_1, \dots, \bar{y}_{k-1})$ and 
    $D = (x \vee \bar{z}_1, \dots, \bar{z}_{k-1})$, and $y_i \ne z_j$ for all $1 \leq i,j \leq k-1$.
    It is clear 
  that $x$ is forced if all $y_i$ come before $x$ {\em or} all $z_i$ come before $x$. Therefore,
  \begin{align*}
     \Pr[\forced(x, \pi) \ | \ \pi(x) = r] & \geq r^{k-1} + r^{k-1} - \Pr[\textnormal{all $y_i$ and all 
     $z_i$ come before $x$}] \\
     & \geq 2\, r^{k-1} - r^{2k-2}  \ . 
  \end{align*}
  On the other hand, by focusing solely on the canonical critical clause tree of $x$,
  we can apply Lemma~\ref{lemma-cut-prob-usual} and conclude that
  \begin{align*}
     \Pr[\forced(x,\pi) \ | \ \pi(x) = r] & \geq Q_{r} - o(1)  \ ,
  \end{align*}
  (we write $Q_r$ instead of $Q_r^{(k)}$ since $k$ is understood), 
  and therefore
  \begin{align*}
     \Pr[\forced(x,\pi)=1] & \geq \int_0^1 \max ( 2\, r^{k-1} - r^{2k-2} , Q_r ) \, dr - o(1) \\
        & = s_k - o(1) + \int_0^1 \max\left(0, 2\, r^{k-1} - r^{2k-2}- Q_r \right) \, dr \ .
  \end{align*}
  It remains to show that the latter term is positive for a substantial range of $r \in [0,1]$.
  We claim that if  $r$ is  sufficiently small, $Q_r$ is only a tiny factor larger than $r^{k-1}$. Indeed,
  From Proposition~\ref{prop-QR-convex}, we know that $Q_r \leq e\, r^{k-1}$, 
  thus $P_r = r \vee Q_r = r + (1-r) e\, r^{k-1} = r (1 + e (1-r) r^{k-2} )$
   and in turn
  $Q_r = P_r^{k-1} \leq \left(r + e\, r^{k-1}\right)^{k-1}
  = r^{k-1} \left(1 + e\, r^{k-2} \right)^{k-1} < r^{k-1} e^{e (k-1) r^{k-2}}$.
  We check that $e^{e (k-1) r^{k-2}}\leq 1.5$ for all $k \geq 3$ and $r \leq 1/16$. Therefore,
  \begin{align*}
   \int_0^1 \max\left(0, 2\, r^{k-1} - r^{2k-2}- Q_r \right) \, dr &  \leq 
   \int_0^{1/16} \left(2\, r^{k-1} - r^{2k-2} - 1.5 \, r^{k-1}\right) \, dr \\
   & = 
   \int_0^{1/16} \left(\frac{1}{2} \, r^{k-1} - r^{2k-2} \right) \, dr \\
  \end{align*}
  and the latter is some positive constant $c_1$ depending only on $k$.
  This concludes the proof of Lemma~\ref{lemma-privileged}.
\end{proof}

\end{document}